\documentclass{amsart}

\usepackage{graphicx} 
\usepackage{comment}
\usepackage{amsmath,amssymb,amsthm,textcomp, tabu}
\theoremstyle{plain}
\usepackage{float}
\usepackage{asypictureB}
\usepackage{setspace}
\usepackage{graphicx}
\usepackage{subcaption}
\newtheorem{theorem}{Theorem}[section]

\newtheorem{remark}[theorem]{Remark}
\newtheorem{lemma}[theorem]{Lemma}

\newtheorem*{statetheorem}{Theorem}
\theoremstyle{definition}

\theoremstyle{remark}

\numberwithin{equation}{section}

\newcommand{\dX}{dX}
\newcommand{\abs}[1]{\lvert#1\rvert}

\newcommand{\C}{\mathbb{C}}
\newcommand{\N}{\mathbb{N}}

\newcommand{\M}{\mathcal{M}}

\newcommand{\ol}{\overline}

\usepackage{xcolor}

\title{Moments of the derivative of the characteristic polynomial of unitary matrices}
\author{E. Alvarez, J.B. Conrey, M.O. Rubinstein and N.C. Snaith }

\begin{document}


\begin{abstract}
    Let $\Lambda_X(s)=\det(I-sX^{\dagger})$ be the characteristic
    polynomial of a Haar distributed unitary matrix $X$. It is believed that the distribution of values of $\Lambda_X(s)$ model the distribution of values of the Riemann zeta-function $\zeta(s)$. This principle motivates many avenues of study. Of particular interest is the behavior of $\Lambda_X'(s)$ and the   distribution of its zeros (all of which lie inside or on the unit circle). In this article we present several identities for the moments 
    of $\Lambda_X'(s)$ averaged over $U(N)$, for $s \in \C$ as well as specialized to $|s|=1$. Additionally, we prove, for positive integer $k$,  that the polynomial
    $\int_{U(N)} |\Lambda_X(1)|^{2k} \dX$ of degree $k^2$ in $N$ divides the polynomial $\int_{U(N)} |\Lambda_X'(1)|^{2k} \dX$ which is of degree $k^2+2k$ in $N$ and that the ratio, $f(N,k)$,  of these moments factors into linear factors modulo $4k-1$ if $4k-1$ is prime. We also discuss the relationship of these moments 
    to a solution of a second order non-linear Painl\'{e}ve differential equation. Finally we give some formulas in terms of the $_3F_2$
    hypergeometric series for the moments in the simplest case when $N=2$, and also study the radial distribution of the zeros of $\Lambda_X'(s)$ in that case.
\end{abstract}
    
\maketitle
 
\section{Introduction}

Ever since work at the turn of the millennium connecting mean values of the Riemann zeta function with averages of characteristic polynomials of unitary matrices selected at random with respect to Haar measure from the unitary group (see for example \cite{kn:keasna00a,kn:cfkrs}), much research has been done on averages of products and ratios of characteristic polynomials and their derivatives over $U(N)$ with Haar measure.  Here we define the characteristic polynomial associated to $X\in U(N)$ to be
\begin{equation}
    \Lambda_X(s)=\prod_{j=1}^N(1-se^{-i\theta_j})=\det(I-sX^{\dagger}),
\end{equation}
where $e^{i\theta_1}, e^{i\theta_2}, \ldots, e^{i\theta_N}$ are the eigenvalues of $X$ and $X^\dagger$ is the conjugate transpose of $X$.
This is not the usual way that we teach students to write the characteristic polynomial, but this way of writing the characteristic polynomial is akin to the Hadamard product of the zeta function.
Furthermore, statistics involving $\Lambda$ are easily related to statistics involving the traditional characteristic polynomial by pulling out exponentials from the product. For the purpose of the moments of the characteristic polynomial or its derivative, both our way and the traditional way of writing the characteristic polynomial lead to identical results.

The origin of the methods used in the current work is \cite{kn:crs06} where  in 2006 the authors introduced a $k$-fold contour integral expression for averages over $U(N)$ with respect to Haar measure of a product of $2k$ characteristic polynomials (reproduced in  Lemma \ref{lem:crs} below).  From this they obtain  an asymptotic formula for large $N$ with integer $k$:
\begin{equation}\label{eq:deriv}
    \int_{U(N)}|\Lambda_X'(1)|^{2k}\dX
    \sim b_k N^{k^2+2k},
\end{equation}
where $dX$ is Haar measure on $U(N)$ normalized so as to be a probability measure
(see~\eqref{eq:haar}), with
\begin{eqnarray}
b_k = (-1)^{k(k+1)/2}
 \sum_{h=0}^k  {k \choose h}
    \bigg(\frac  d {dx}\bigg)^{k+h} \bigg(e^{-x} x^{-k^2/2}\det_{k\times k}
\big( I_{i+j-1}(2\sqrt{x}) \bigg)\bigg|_{ x=0},
\end{eqnarray}
and $I_\nu(z)$ denotes the modified Bessel function of the first kind.
The logarithmic derivative of the above determinant was shown by Forrester and Witte \cite{kn:forwit06} to be a solution of a Painlev{\'e} differential equation.

In this current work we derive several identities for the $2k$-th moments of $|\Lambda_X'(x)|$ for complex $x$ and also specialized to $|x|=1$, i.e. on the unit circle. In the last section, also study the radial distribution of the zeros of $\Lambda_X'$.

Our first theorem, in Section \ref{sec:2}, gives our first 
formula for the moments of $\abs{\Lambda_X'(x)}$, for $x$ on the unit circle.  Importantly, this is an exact expression for finite $N$. 

Having exact formulae for finite $N$, rather than only leading asymptotic expansions, is valuable in that they can provide a model for the lower terms in the analogous statistics for the moments of the derivative of the Riemann zeta function, and can also reveal combinatorial identities that underpin the moments. While our focus here is on random matrices, in other contexts, such as the full asymptotic expansion of
the moments of $\zeta$, having a precise conjecture for the lower terms was crucial
in testing the conjectured number theoretic moments~\cite{kn:cfkrs}. 

\begin{statetheorem}[Theorem \ref{theo:sumofdets} below]  For $k$ a non-negative integer, and $x \in \C$ with $|x|=1$,
\begin{eqnarray}
   \label{eq:1st}
   &&  \int_{U(N)} \abs{\Lambda_X'(x)}^{2k} \dX \nonumber \\
    && = (-1)^{\binom{k}{2}} \sum_{m=0}^k \binom{k}{m}N^{k-m} (-1)^m  \sum_{\substack{\sum_{j=1}^k t_{j} = k+m}}\binom{k+m}{t_{1}, \dots, t_{k}}    \nonumber \\ && \times 
    \det\left[ \binom{N+k+i+j-2}{2k+t_j -1} \right]_{\substack{1 \leq i \leq k \\ 1 \leq j \leq k}}. 
\end{eqnarray}
\end{statetheorem}
Note that the right hand side does not depend on the argument of $x$ 
since Haar measure on $U(N)$ is invariant under rotation.
For instance, looking ahead to~\eqref{eq:N integral},
letting $x=r \exp(i\theta)$,
and changing variables $\omega_j=\theta_j-\theta$, gives the $2k$-th moment of $|\Lambda'(r)|$. 

While computationally challenging, (\ref{eq:1st}) is explicit and simple enough
to work out the first few moments, as functions of $N$,
by summing the terms.

Additionally, by studying the determinants in the inner sum, we
are able to prove the following theorem which shows a connection to the moments of $|\Lambda_X(1)|$.
\begin{statetheorem}[Theorem \ref{thm1} below]
For $k$ a non-negative integer, 
\begin{align}
    \int_{U(N)} \abs{\Lambda_X'(1)}^{2k}\dX
    =
    \int_{U(N)} \abs{\Lambda_X(1)}^{2k}\dX \times f(N,k), 
\end{align} where $f(N,k)$ is a polynomial in $N$ of degree $2k$. 
\end{statetheorem} 
Note that the $2k$-th moment of $|\Lambda_X(1)|$ is 
known explicitly~\cite{kn:keasna00a}. They prove:
\begin{eqnarray}
    \label{eq:ks product}
    \int_{U(N)} \abs{\Lambda_X(1)}^{2k}\dX
    =
    \prod_{j=1}^N \frac{\Gamma(j)\Gamma(2k+j)}{\Gamma(k+j)^2}
    =
    \prod_{j=0}^{k-1}
    \left(
        \frac{j!}{(j+k)!}
        \prod_{i=0}^{k-1} (N+i+j+1)
    \right).
\end{eqnarray}
The first equality is valid for  $\Re{k} > -1/2$ while the second is for non-negative integer $k$. The authors of~\cite{kn:agkw} observed a factorization property for general joint moments,
where the joint moment factors into the moment of the characteristic polynomial itself and
into a factor that is connected to Painlev{\'e} equations.

A related problem, that of moments of moments, yields interesting polynomials that can be understood through lattice point counting in polytopes \cite{kn:B} \cite{kn:aen23}. It would be interesting to see if our polynomials $F(N,k)$ can be similarly understood via counting points.

Equation~\eqref{eq:1st} is also explicit enough to discover a curious factorization (see Section \ref{sect:modulo}) in $\mathbb{Z}_p[N]$, if $p=4k-1$ is prime.
We discovered this theorem empirically by examining the first few examples.

\begin{statetheorem}[Theorem \ref{theo:mod4k-1} below]
For $k$ a non-negative integer, if $4k-1$ is prime
\begin{eqnarray}
    &&(4k-1) \int_{U(N)}|\Lambda_X'(1) |^{2k} \dX=  \notag \\
    &&=  \frac{(-2) (N-2k+1) (N-2k+2)\cdots N}{ (k-1)! (k-1)!}\int_{U(N)} \abs{\Lambda_X(1)}^{2k}\dX \mod 4k-1.
\end{eqnarray}
\end{statetheorem}
Note, as part of our proof, it will emerge that
that the rational coefficients of powers of $N$ of $\int_{U(N)}|\Lambda_X'(1) |^{2k} dX$ have a single power of $4k-1$
in their denominators. In interpreting the left hand side, the factor $4k-1$ on the left hand side should first be cancelled with the $4k-1$
in the denominators ahead of reducing $\mod 4k-1$.
Apart from that, all arithmetic in this expression is modulo $4k-1$. So all other integers appearing in denominators are to be interpreted as inverses mod $4k-1$.

In Section \ref{sect:three}, we derive formulas for the $2k$-th moment
of $|\Lambda_X'(x)|$ for any $x \in \C$.
\begin{statetheorem}[Theorem \ref{thm:t1 t2} below]
   Let $x \in \C$, and $k$ be a non-negative integer. Then
\begin{equation}
   \int_{U(N)} \abs{\Lambda_X'(x)}^{2k} \dX \nonumber \\
    =(-1)^\frac{(k+1)k}{2} 
    \frac{d^k}{dt_1^k}
    \frac{d^k}{dt_2^k} 
    e^{-t_1 N}
    \det\left[ F_{N+k+i+j-1,k}(t_1,t_2,x) \right]_{\substack{1 \leq i \leq k \\ 1 \leq j \leq k}} \bigg|_{t_1=t_2=0} 
\end{equation}
where
\begin{equation}
    F_{a,k}(t_1,t_2,x)
    =
    \frac{1}{2\pi i} \oint \frac{w^{a-1}}{(w-1)^k (w-|x|^2)^k}
    \exp\left(t_1/(w-1) + t_2/(w-|x|^2)\right) dw,
\end{equation}
and the contour is a circle centred on the origin enclosing the points $1$ and $|x|^2$.
To clarify, the right hand side of the above is evaluated,
after carrying out the derivatives, at $t_1=t_2=0$.
Furthermore, if $|x| \neq 1$ and $a$ is a positive integer, then $F_{a,k}(t_1,t_2,x)$ is also equal to
\begin{align}
    \sum_{0 \leq m+n+2k \leq a}
        \frac{t_1^m}{m!}\frac{t_2^n}{n!} &\Biggl(
            \frac{|x|^{2(a-n-k)}}{(|x|^2-1)^{m+k}}
            \sum_{l=0}^{n+k-1} {a-1 \choose n+k-1-l} { -m -k \choose l} \frac{|x|^{2l}}{(|x|^2-1)^{l}} \notag \\
            &+
            \frac{1}{(1-|x|^2)^{n+k}}
            \sum_{l=0}^{m+k-1} {a-1 \choose m+k-1-l} { -n -k \choose l} \frac{1}{(1-|x|^2)^{l}}
    \Biggr).
\end{align}
\end{statetheorem}
Here, we use the usual formula for binomial coefficients, including when the first argument is negative:
${M \choose l} = M(M-1)\ldots(M-l+1)/l!$, for all $l,M\in \mathbb{Z}$,
$l\geq 0$,  with the numerator taken to be 1 if $l=0$.

\begin{statetheorem}[Theorem \ref{thm:x=1} below]
Let $|x|=1$, and $k$ be a non-negative integer. Then
\begin{eqnarray}
    &&
    \int_{U(N)} \abs{\Lambda_X'(x)}^{2k} \dX \nonumber \\
    && \qquad = 
    (-1)^{k} \sum_{h=0}^k {k \choose h} N^{k-h} (d/dt)^{k+h}
    \det_{k\times k} \left[L_{N+i-j}^{(2k-1)}(t) \right]
    \bigg|_{t=0}.
\end{eqnarray}
\end{statetheorem}
Related formulae featuring this determinant first appeared in~\cite{kn:winn12}
and, later, in~\cite{kn:basor_et_al18}, but for the characteristic polynomial
$\Lambda$ rotated so as to be real on the unit circle.

Curiously, if instead of proceeding via the $k$-fold contour integral of \cite{kn:crs06} we start with the Weyl integration formula for Haar measure on $U(N)$ (an $N$-fold integral) and proceed with similar steps, a closely related expression appears, but featuring an $N\times N$ determinant instead of $k\times k$, exposing a duality between the parameters $k$ and $N$.

\begin{statetheorem}[Theorem \ref{thm:x=1 N version} below]
For positive integer $k$ and $|x|=1$
\begin{eqnarray}
    &&
    \int_{U(N)} \abs{\Lambda_X'(1)}^{2k} \dX \nonumber \\
    &&\qquad  =
    (-1)^{kN} \sum_{h=0}^k {k \choose h} (-1)^{h} N^{k-h} (d/dt)^{k+h}
    \det_{N\times N} \left[L_{k+i-j}^{(-2k-1)}(t) \right]
    \bigg|_{t=0}.
\end{eqnarray}
\end{statetheorem}

We can obtain more explicit expressions, see Section \ref{sect:four}, if we restrict to $N=2$: 

\begin{statetheorem}[Theorem \ref{thm:B1} below] For positive integer $k$ and any complex $x$ we have
\begin{eqnarray*}
\int_{U(2)} |\Lambda_X'(x)|^{2k} ~\dX&=&
\sum_{m=0}^k \frac{\binom{k}{m}^2 (4\abs{x}^2)^m\binom{2k-2m}{k-m}}{k-m+1}\\
&=&
\frac{\binom{2k}{k}}{k+1} {}_3F_2(-1-k,-k,-k;1,\frac 12 -k; |x|^2).
\end{eqnarray*}
\end{statetheorem}
Finally,  allowing $k$ to be non-integer,

\begin{statetheorem}[Theorem \ref{thm:b2} below]
For all $k \in \C$, and $x \in \C$ with $|x|>1$ we have
   \begin{eqnarray*}
   \int_{U(2)} |\Lambda_X'(x)|^{2k} dX
   =
   2^{2 k} |x|^{2 k} \, _3F_2\left(\frac{1}{2},-k,-k;1,2;|x|^{-2}\right).
\end{eqnarray*}
Additionally, this formula extends to $|x|=1$ if $\Re k > -1$.

\end{statetheorem}
Interestingly, when $k$ is not an integer, numerical investigation shows that these expressions do not agree, see Section \ref{sect:four}.

It would be worthwhile to investigate and understand the complex moments of $\Lambda_X'$ for finite, but we have not done so here.
A formula was given by Simm and Wei in~\cite{kn:simmwei25} for the leading asymptotic of complex moments, as $N \to \infty$, for $|x|<1$.

Finally, we end with a theorem for the logarithmic  average of $\Lambda'_X(r)$ for $N=2$. 
\begin{statetheorem}[Theorem \ref{thm:log lambda-prime} below]
For $0 \leq r<1$ we have
  \begin{eqnarray*}  \int_{U(2)} \log|\Lambda_X'(r)| ~dX
  =\frac{2 r \,
   _3F_2\left(\frac{1}{2},\frac{1}{2},\frac{1}{2};\frac{3}{2},\frac{3}{2};r^2\right)
   +r \sqrt{1-r^2} +\sin ^{-1}(r)}{\pi }-\frac{1}{2}.
    \end{eqnarray*}
\end{statetheorem}

The history of these questions in random matrix theory originated  in the thesis of Hughes \cite{kn:hug01} and also includes derivatives of a closely related function (the analogue of the Hardy's $Z$-function in number theory) which is real on the unit circle and defined by
\begin{equation}
    Z_X(s)=e^{i\pi N/2}e^{i\sum_{n=1}^N\theta_n/2} s^{-N/2}\Lambda_X(s).
\end{equation}

Hughes (although with slightly different notation) looked at moments of the form (note that for $k=h$ this reduces to (\ref{eq:deriv}))
\begin{equation}\label{eq:Lderiv}
    \int_{U(N)}|\Lambda_X(1)|^{2k-2h}|\Lambda_X'(1)|^{2h} \dX
\end{equation}
and
\begin{equation} \label{eq:Zderiv}
    F_N(h,k):= \int_{U(N)}|Z_X(1)|^{2k-2h}|Z_X'(1)|^{2h} \dX,
\end{equation}
with $h>-1/2$ and $k>h-1/2$.  Hughes shows that the limit 
\begin{equation}\label{eq:Flimit}
    \lim_{N\rightarrow \infty} \frac{1}{N^{k^2+2h}}F_N(h,k)=F(h,k)
\end{equation}
exists and conjectured a form for $F(h,k)$ that continues to non-integer $k$ but not non-integer $h$. This work was extended and proved by Dehaye \cite{kn:dehaye10}. The method of \cite{kn:crs06} was extended in \cite{kn:bbbcprs} to the mixed moment (\ref{eq:Lderiv}), including the connection to a Painlev{\'e} equation via a determinant of Bessel functions.  The first step towards anything other than an even integer exponent on the derivative of the characteristic polynomial is the work of Winn \cite{kn:winn12}, who obtained a concrete formula for $F_N(h,k)$ when $h=(2m-1)/2$ for $m\in\mathbb{N}$.  Assiotis, Keating and Warren \cite{kn:akw22}, however, establish that the limit (\ref{eq:Flimit}) exists for real $h$ and $k$ and relate the leading order coefficient to the expectation value of a particular random variable.

The papers~\cite{kn:abgs}~\cite{kn:Forrester22} compute the limiting moments $F(h,k)$ of~\eqref{eq:Flimit} in the case $k \in \mathbb{Z}_{\geq 0}$
and $h$ complex (with restrictions so that the moments exist).
However, their work does not apply here if one wishes to obtain the complex moments of $\Lambda_X'$ alone, since one cannot set $h=k$ if $k$ is an integer and $h$ is a non-integer complex number.

Averages of higher derivatives, quantities of the form of (\ref{eq:Zderiv}), but instead of the zeroth and first derivative, with two arbitrary orders of differentiation, $n_1$ and $n_2$, are addressed by Keating and Wei for integer $k$ and $h$ in \cite{kn:keawei1,kn:keawei2} where they find the moment is asymptotically, for large $N$, of order $N^{k^2+2(k-h)n_1+2hn_2}$ and they give explicit forms for the leading order coefficients. Barhoumi-Andr{\'e}ani \cite{kn:barhoumi} addresses combinations of more than two different derivatives, but still asymptotically and always evaluated at the point 1 and expresses the leading order coefficient as a multiple contour integral. 

Since the first release of the current work, there has been further development, where Assiotis, Gunes, Keating and Wei \cite{kn:agkw} have established the leading order behaviour for any number of derivatives of any order, of either $\Lambda$ or $Z$, evaluated at the point 1,  importantly with real exponents, thus significantly extending the work described in the previous paragraph. They relate the leading order coefficient to the expectation value of a random variable. 
Additionally, they give a recursive method that allows one to
explicitly compute, for fixed matrix size $N$, arbitrary joint moments of
derivatives of any order, with integer exponents. But carrying out their
recursion to obtain a closed form formula is difficult and not done in their
paper.

There are a wealth of interesting results concerning the average (\ref{eq:Zderiv}) in \cite{kn:basor_et_al18}. They express $F_N(h,k)$, for integer $h$ and $k$, in terms of the same determinant of Laguerre polynomials that arises in our Theorem~\ref{thm:x=1}, giving a result for finite $N$.
They relate the logarithmic derivative of this determinant to the solution of a Painlev{\'e} differential equation. This allows them to solve recursively for $F_N(h,k)$ resulting in expressions that continue to non-integer $k$. They explicitly write out the first few $F_N(h,k)$ for small integer $h$ values and general $k$, in the form $F_N(0,k)$ multiplied by what looks like a polynomial in $N$. 

In Section \ref{subsection:diffeq} we describe the differential equation and its use in our setting
for the fast calculation of (\ref{eq:deriv}) for specific values of $k$.

There is also literature studying average values of the logarithmic derivative of the characteristic polynomial of $X\in U(N)$,
$\Lambda_X'(s)/\Lambda_X(s)$.
Whereas most of the results in the literature for moments (as opposed to ratios with characteristic polynomials in the denominator) including derivatives of characteristic polynomials evaluate the derivative at the point 1,  
averages of
$\Lambda_X'(s)/\Lambda_X(s)$
must be evaluated away from the unit circle ($|s|\neq 1$) due to singularities at the eigenvalues of $X$ from the characteristic polynomial in the denominator, see \cite{kn:bbbcprs} and \cite{kn:alvsna20}.  

One of the motivations for studying the moments of derivatives of characteristic polynomials of random matrices drawn from $U(N)$ with Haar measure is because from the moments information can be retrieved about the distribution of the zeros of derivatives of characteristic polynomials.  With the analogy relating the characteristic polynomial to the Riemann zeta function, this has the potential to shed light on the distribution of zeros of the derivative of the zeta function \cite{kn:mezzadri03,kn:dffhmp} and hence on the Riemann Hypothesis by the ideas of Levinson \cite{kn:levinson,kn:conrey89}.  In Section \ref{sect:zerodist} this motivates an investigation into the radial distribution of the zeros of the derivative of the characteristic polynomial through moments of the derivative evaluated on the interior of the unit circle. 

All of the above has prompted similar calculations in other ensembles of unitary matrices. For matrices in the classical compact groups $SO(2N)$, $SO(2N+1)$ and $USp(2N)$, Snaith and Alvarez \cite{kn:alvsna20} consider asymptotic formulae for moments of the logarithmic derivative, where the exponent is an integer, evaluated at a point on the real axis approaching 1 faster than $1/N$ as the matrix size $N$ goes to infinity. This is extended to non-integer exponent for $SO(2N+1)$ in \cite{kn:alvbousna} and by Ge \cite{Ge4} to the other classical compact groups.  Asymptotic formulae for low derivatives of characteristic polynomials evaluated at the point 1 can be found in \cite{kn:abprw14}, and for joint higher derivatives (in analogy to Keating and Wei) in \cite{kn:andbes}. 
Mixed moments of characteristic polynomials and their derivatives for the
CUE, circular $\beta$-ensemble, and circular Jacobi $\beta$-ensemble are
studied, respectively, in \cite{kn:abgs} \cite{kn:Forrester22} \cite{kn:ags22}.

\section{A finite-$N$ determinant formula for moments of the derivative $\Lambda'(1)$}
\label{sec:2}
In this section we will prove the following theorem on derivatives of characteristic polynomials over $U(N)$ with Haar measure. The result is exact for finite matrix size $N$ and, as can be seen from the review of the literature above, there currently exist few results that give more than just the leading order behaviour in $N$. 
\begin{theorem} \label{theo:sumofdets} For $k$ a positive integer, 
\begin{eqnarray}
   &&  \int_{U(N)} \abs{\Lambda_X'(1)}^{2k} \dX \nonumber \\
    && = (-1)^{\binom{k}{2}} \sum_{m=0}^k \binom{k}{m}N^{k-m} (-1)^m  \sum_{\substack{\sum_{j=1}^k t_{j} = k+m}}\binom{k+m}{t_{1}, \dots, t_{k}}    \nonumber \\ && \times 
    \det\left[ \binom{N+k+i+j-2}{2k+t_j -1} \right]_{\substack{1 \leq i \leq k \\ 1 \leq j \leq k}}. 
\end{eqnarray}
\end{theorem}
\begin{proof}[{\bf Proof of Theorem \ref{theo:sumofdets}}]
The proof starts with a variant of Lemma 3 from \cite{kn:crs06}.
\begin{lemma}\label{lem:crs}
\begin{eqnarray}\label{eq:lotsofa}
    &&\int_{U(N)}\prod_{j=1}^k \Lambda_X(1/a_j)\Lambda_{X^{\dagger}}(a_{j+k})\dX\nonumber \\
    &=& \frac{1}{k!(2\pi i)^k} \oint \prod_{j=1}^k \left(\frac{u_j}{a_j}\right)^N \prod_{\substack{1\leq i \leq k\\ 1\leq j \leq 2k}} \frac{1}{1-\frac{a_j}{u_i}}\prod_{i\neq j} \left( 1-\frac{u_j}{u_i} \right) \prod_{j=1}^k\frac{1}{u_j}du_j \nonumber \\
    &=& \frac{(-1)^{\binom{k}{2}}}{k!(2\pi i)^k} \oint \prod_{j=1}^k \frac{u_j^{N-k}}{a_j^N} \prod_{\substack{1\leq i\leq k\\1\leq j\leq 2k}} \frac{u_i}{u_i-a_j} \prod_{i<j}(u_j-u_i)^2 \prod_{j=1}^k du_j.
\end{eqnarray}
with each of the $k$ contours in this $k$-dimensional integral being simple closed contours enclosing all the poles of the integrand.
\end{lemma}
The precise contours are not important since
we will eventually compute our integrals via residues, but one can take them to be sufficiently large circles centred on the origin.
The proof of this Lemma is the same as in~\cite{kn:crs06}. The difference here is that we are stating the Lemma with the variable $u$ which corresponds to their $e^w$. This makes no difference if one restricts to sufficiently small $a_j$, as they do, by eventually setting all the $a_j=1$. However, in Section~\ref{sect:three}, we will need to allow the $a_j$'s to be more general, and
we would run into complications if we were to use $e^w$, due to the periodic nature of the exponential function counting the same poles repeatedly. Hence we have stated our Lemma in this form.

The idea is the same as previous work on moments of derivatives of characteristic polynomials, or joint moments of derivatives and the characteristic polynomial itself that start from Lemma \ref{lem:crs} (for example  \cite{kn:crs06,kn:bbbcprs,kn:keawei1,kn:keawei2}) with the difference that we do not scale variables with $1/N$ and make a large $N$ approximation.  This allows us to calculate exactly for finite $N$. We will take one derivative with respect to each $a_j$, then set them all equal to 1, thus finding the $2k^{th}$ moment of the derivative of the characteristic polynomial evaluated at 1. After taking the derivative, we have a product of two multinomial expansions, so we introduce two new parameters, $t$ and $s$ to factorise those terms as the derivative of an exponential. We then apply Andr\'eief's identity to express the multiple contour integral as a determinant. For simplicity, we now replace $a_j$ with $1/a_j$ for $1 \leq j \leq k$. These first $k$ of the $a_j$ appear in the factors
\begin{equation}\label{preder}
    a_j^N \prod_{1 \leq i \leq k} \frac{u_i}{u_i - \frac{1}{a_j}} = a_j^N \prod_{1 \leq i \leq k} \frac{u_ia_j}{u_ia_j - 1} 
\end{equation} whose derivative with respect to $a_j$ is 
\begin{align}
    Na_j^{N-1}\prod_{1 \leq i \leq k} \frac{u_ia_j}{u_ia_j - 1} + a_j^N \sum_{n=1}^k \frac{-u_n}{(u_na_j-1)^2}\prod_{m\neq n} \frac{u_ma_j}{u_ma_j-1}.
\end{align} Evaluated at $a_j=1$, each of these $k$ factors gives
\begin{equation}
    \label{eq:postder}
    \prod_{1\leq i \leq k} \frac{u_i}{u_i -1} \left( N - \sum_{m=1}^k \frac{1}{u_m-1}\right).
\end{equation} Next, we look at the derivatives with respect to $a_j$ for $k < j \leq 2k$. These $a_j$s appear only in the factors
\begin{equation}
    \prod_{1 \leq i \leq k} \frac{u_i}{u_i - a_j}
\end{equation} and each of these derivatives is given by 
\begin{align}
\sum_{n=1}^k \frac{u_n}{(u_n-a_j)^2}\prod_{m \neq n}\frac{u_m}{u_m - a_j}
\end{align} which, evaluated at $a_j = 1$, yields
\begin{equation}
    \label{eq:postder2}
    \prod_{1 \leq i \leq k} \frac{u_i}{u_i-1}\left( \sum_{m=1}^k \frac{1}{u_m-1} \right).
\end{equation}
We combine all $2k$ derivatives to get
\begin{eqnarray}
    &&\int_{U(N)} \abs{ \Lambda'_X(1)}^{2k} \dX  \nonumber \\
    &&\qquad= \frac{(-1)^{\binom{k}{2}}}{k!(2\pi i)^k}\oint \left( \prod_{j=1}^k u_j^{N-k} \left( \frac{u_j}{u_j-1} \right)^{2k} \right) \left( N- \sum_{m=1}^k \frac{1}{u_m-1}\right)^k \left(\sum_{m=1}^k \frac{1}{u_m-1} \right)^k \Delta^2({\bf u}) \prod_{j=1}^k du_j \label{alt}\nonumber\\
    &&\qquad= \frac{d^k}{dt^k} \frac{d^k}{ds^k} \frac{(-1)^{\binom{k}{2}}}{k!(2\pi i)^k}\oint \left( \prod_{j=1}^k  \frac{u_j^{N+k}}{(u_j-1)^{2k}}\right) \exp \left( Nt- (t-s)\sum_{m=1}^k \frac{1}{u_m-1} \right) \Delta^2({\bf u}) \prod_{j=1}^k du_j \Bigg|_{t=s= 0}
    , \label{eq:s and t}
\end{eqnarray} 
where the Vandermonde determinant is
 \begin{equation}
    \label{absvan}
    \Delta({\bf u}) = 
    \Delta(u_1,\ldots,u_k)
    :=
    \prod_{1\le i<j \le k}(u_j - u_i)
    =
    \det_{k \times k} \left[ u_j^{i-1}\right].
\end{equation}
The last equality in~\eqref{eq:s and t} uses
\begin{equation}
    \frac{d}{dt}\exp\left( Nt-t\sum_{m=1}^k\frac{1}{u_m-1}\right)\Bigg|_{t=0}=\left(N-\sum_{m=1}^k\frac{1}{u_m-1}\right)
\end{equation}
and is designed to move the $k$-th powers involving
the sum over $m$ into an exponent so that the variables can be  separated.

A form of Andr\'eief's identity says
\begin{eqnarray}
    \label{Aint}
    &&\frac{1}{n!} \int_{J^n} \left(\prod_{j=1}^n f(u_j)\right) \det_{n
    \times n}(\psi_i(u_j)) \det_{n \times n} (\phi_i(u_j)) d u_1
    \cdots d u_n\nonumber \\
    &&=\frac{1}{n!} \int_{J^n} \det_{n \times
    n}\left(f(u_j)\psi_i(u_j)\right)
    \det_{n \times n} (\phi_i(u_j)) d u_1 \cdots d u_n
    = \det_{n \times n} \left( \int_J f(u)\psi_i(u) \phi_j(u) du
    \right),
\end{eqnarray}
for some interval or contour $J$ and a sequence of functions $\phi$ and $\psi$. 

With the introduction of the variables $s$ and $t$ in (\ref{eq:s and t}), we have worked the integrand multiplying the Vandermonde determinants into the multiplicative form $\prod_{j=1}^kf(u_j)$,
so applying Andr\'eief, we now have
\begin{eqnarray}
   && \int_{U(N)} \abs{ \Lambda'_X(1)}^{2k} \dX  \nonumber \\
    && = (-1)^{\binom{k}{2}}\frac{d^k}{dt^k} \frac{d^k}{ds^k} \exp(Nt) \det\left[ \frac{1}{2\pi i} \oint \frac{u^{N+k+i+j-2}}{(u-1)^{2k}}\exp\left(\frac{s-t}{u-1}\right) du\right]_{\substack{1 \leq i \leq k \\ 1 \leq j \leq k}}\Bigg|_{t=s=0}.
\end{eqnarray}
Next, we apply the product rule to the differentiation with respect to $t$, on account of the extra $\exp(Nt)$ in front of the determinant, and the above becomes
\begin{eqnarray}
    && \int_{U(N)} \abs{ \Lambda'_X(1)}^{2k} \dX =(-1)^{\binom{k}{2}}\sum_{m=0}^k
    \binom{k}{m}
    N^{k-m}
    \nonumber \\
    &&\qquad\times\left(\frac{d^m}{dt^m}\frac{d^k}{ds^k}
    \det\left[ \frac{1}{2\pi i} \oint \frac{u^{N+k+i+j-2}}{(u-1)^{2k}}\exp\left(\frac{s-t}{u-1}\right) du\right]_{\substack{1 \leq i \leq k \\ 1 \leq j \leq k}}\right)\Bigg|_{t=s=0}\nonumber \\
    &&=(-1)^{\binom{k}{2}}\sum_{m=0}^k
    \binom{k}{m}
    N^{k-m} \nonumber \\
    &&\qquad\times(-1)^m \left(\frac{d^{k+m}}{dt^{k+m}}
    \det\left[ \frac{1}{2\pi i} \oint \frac{u^{N+k+i+j-2}}{(u-1)^{2k}}\exp\left(\frac{t}{u-1}\right) du\right]_{\substack{1 \leq i \leq k \\ 1 \leq j \leq k}}\right)\Bigg|_{t=0},
    \label{eq:moment as derivative of det}
\end{eqnarray}
where in the final line we have combined the $s$ and $t$ derivatives, noting that they have exactly the same effect on the integrand, collecting the extra minus signs from the $t$ differentiations in the factor $(-1)^m$.

Finally, when differentiating a determinant one time with respect to the variable $t$ we get, by the product rule, a sum of determinants
\begin{equation}\label{detid}
    \frac{d}{dt} \det(a_1, a_2, \dots, a_n) = \det\left(\frac{d}{dt}a_1, a_2, \dots, a_n\right) + \dots + \det\left(a_1, a_2, \dots, \frac{d}{dt} a_n\right).
\end{equation}
Here, the $a_j$'s are column vectors.
Note that, by writing $u=(u-1)+1$ in the numerator and using the binomial expansion in the middle expression (or else by Cauchy's Integral Formula for derivatives), we compute the residue
\begin{equation}
    \frac{d}{dt}\frac{1}{2\pi i} \oint \frac{u^{N+k+i+j-2}}{(u-1)^{2k}}\exp\left(\frac{t}{u-1}\right) du\Bigg|_{t=0}=\frac{1}{2\pi i} \oint \frac{u^{N+k+i+j-2}}{(u-1)^{2k+1}}du=\binom{N+k+i+j-2}{2k}.
\end{equation}
Repeating this process, differentiating $k+m$ times with respect to $t$, results in a sum involving many determinants. We index according to how many times we have differentiated the $\ell$-th column with respect to $t$,  say $t_\ell \geq 0$, with $\sum t_\ell = k+m$. This results in
the following equation.
\begin{eqnarray}
    && \int_{U(N)} \abs{ \Lambda'_X(1)}^{2k} \dX =(-1)^{\binom{k}{2}}
    \sum_{m=0}^k \binom{k}{m}N^{k-m}
    (-1)^m
    \sum_{\substack{\sum_{\ell=1}^k t_\ell = k+m}}
    \binom{k+m}{t_1, \dots , t_k}
    \nonumber \\
    && \qquad \times \det\left[ \frac{1}{2\pi i}
    \oint \frac{u^{N+k+i+j-2}}{(u-1)^{2k+t_j}} du\right]_{\substack{1 \leq i \leq k \\ 1 \leq j \leq k}}\nonumber \\
    &&=(-1)^{\binom{k}{2}}
    \sum_{m=0}^k \binom{k}{m}N^{k-m}
    (-1)^m
    \sum_{\substack{\sum_{\ell=1}^k t_\ell = k+m}}
    \binom{k+m}{t_1, \dots, t_k}
    \nonumber \\
    && \qquad \times \det\left[ \binom{N+k+i+j-2}{2k+t_j -1} \right]_{\substack{1 \leq i \leq k \\ 1 \leq j \leq k}}. \label{sumofdets1}
\end{eqnarray}
The multinomial coefficient arises because we can arrive at the same number of differentiations of each column in multiple ways by doing the differentiations in different orders, eg. first differentiating the first column and then the second, versus differentiating the second and then the first.  So in the sum arising from applying (\ref{detid}) $k+m$ times, we group all terms with the same number of differentiations of column 1, column 2, etc, and account for the multiplicity using the multinomial coefficient.  
\end{proof}

\subsection{Factoring out moments of the characteristic polynomial}

Using equation \eqref{sumofdets1}, we can obtain the exact moment values for small $k$. Doing so, we can clearly see that the $2k^{th}$ moment of
$\abs{\Lambda_X'(1)}$
 factors into two parts: the $2k^{th}$ moment of $\abs{\Lambda_X(1)}$, i.e. the non-differentiated characteristic polynomial, as well as some polynomial in $N$ of degree $2k$.
We therefore define
\begin{equation}
    f(N,k) = 
    \frac{
        \int_{U(N)} \abs{\Lambda_X'(1)}^{2k}\dX
    }
    {
        \int_{U(N)} \abs{\Lambda_X(1)}^{2k}\dX
    },
\end{equation}
and list these functions, below, for $k \leq 6$:
\begin{align}
    f(N,1) &=
        N (1 + 2 N)/6
    \notag \\
    f(N,2) &=
        N (12 + 27 N + 40 N^2 + 61 N^3)/840
    \notag \\
    f(N,3) &=
        N (840 + 2174 N + 2829 N^2 + 2980 N^3 + 3933 N^4 + 6648 N^5)/388080
    \notag \\
    f(N,4) &=
        N \left( 211680 + 605724 N + 828464 N^2 + 835627 N^3 + 831344 N^4   \right. \notag \\
         & \quad \quad \left.  + 915970 N^5 + 1279520 N^6 + 2275447 N^7 \right)/544864320
     \notag \\
     f(N,5) &=
         N
         \left( 544864320 + 1680129432 N + 2440884600 N^2 + 2498415180 N^3 + 2320167235 N^4 + 2266635142 N^5 \right. \notag \\
         &
         \quad \quad \left. + 2448916150 N^6 + 2872062460 N^7 + 4060136575 N^8 + 7401505546 N^9 \right)/7190496593280
     \notag \\
     f(N,6) &= N \left( 222615993600 + 727617496320 N + 1115985182112 N^2  + 1176700689444 N^3 + 1073389052700 N^4   \right. \notag\\
      & \quad \quad \left. + 988586333095 N^5  + 978075305136 N^6 + 1034426527167 N^7 +  1167375408300 N^8 \right. \notag \\
      & \quad \quad  \left. + 1398326972685 N^9 + 1974154070952 N^{10} + 3654712923689 N^{11}\right)/14333056542604800 \notag \\
\end{align}
We will prove the following:
\begin{theorem}\label{thm1}
For $k\in \N$, $f(N,k)$ is a polynomial in $N$ of degree $2k$. 
\end{theorem} 
The data above together with the plot of the roots below in Figure \ref{fig:roots} suggests that $f(N,k)$ may be irreducible over the rationals.
\begin{proof}[{\bf Proof of Theorem \ref{thm1}}]
As in the proof of Theorem \ref{theo:sumofdets},
we start with the average of the characteristic polynomial itself,  letting $a_j = 1$ for all $j$ in \eqref{eq:lotsofa}, we have, using the methods of the previous section, 
\begin{eqnarray}
&&    \int_{U(N)} \abs{\Lambda_X(1)}^{2k}\dX = \frac{(-1)^{\binom{k}{2}}}{k!(2\pi i)^k}\oint \prod_{j=1}^k u_j^{N-k} \prod_{\substack{1 \leq i \leq k}} \frac{u_i^{2k}}{(u_i-1)^{2k}} \prod_{i < j} (u_j-u_i)^2 \prod_{j=1}^kdu_j \nonumber \\
&&    = (-1)^{\binom{k}{2}} \det_{k\times k} \left[ \frac{1}{2\pi i} \oint \frac{u^{N+k+i+j-2}}{(u-1)^{2k}} du \right]\nonumber \\
 &&   = (-1)^{\binom{k}{2}} \det_{k\times k} \left[ \binom{N+k+i+j-2}{2k -1} \right]. \label{termtofactor}
\end{eqnarray} 

The quantity on the left hand side is the familiar $2k$th moment of the characteristic polynomial, many forms for which are well-known, but we continue to evaluate it using the above determinant as the method will help us complete this proof of Theorem \ref{thm1}. First, we note how one can compute this determinant simply by factoring out common terms and using row reduction. Indeed, 
\begin{eqnarray}
&& \det \left[ \binom{N+k+i+j-2}{2k -1} \right] = \det
\begin{bmatrix}
\binom{N+k}{2k-1} & \binom{N+k+1}{2k-1} & \binom{N+k+2}{2k-1} & \dots & \binom{N+2k-1}{2k-1} \\
\binom{N+k+1}{2k-1} & \binom{N+k+2}{2k-1} & \binom{N+k+3}{2k-1} & \dots & \binom{N+2k}{2k-1} \\
\binom{N+k+2}{2k-1} & \binom{N+k+3}{2k-1} & \binom{N+k+4}{2k-1} & \dots & \binom{N+2k+1}{2k-1} \\
\vdots & \vdots & \vdots & \vdots & \vdots \\
\binom{N+2k-1}{2k-1} & \binom{N+2k}{2k-1} & \binom{N+2k+2}{2k-1} & \dots & \binom{N+3k-2}{2k-1} 
\end{bmatrix} \label{factdet1}\\
&& = \frac{1}{((2k-1)!)^k} \times \\
&& \det 
\begin{bmatrix}
(N+k)\dots(N-k+2) & (N+k+1)\dots(N-k+3) & \dots & (N+2k-1) \dots (N+1) \\
(N+k+1)\dots(N-k+3) & (N+k+2)\dots(N-k+4) & \dots & (N+2k) \dots (N+2) \\
\vdots & \vdots & \vdots & \vdots \\
(N+2k-1)\dots(N+1) & (N+2k)\dots (N+2) & \dots & (N+3k-2)\dots (N+k)
\end{bmatrix}. \nonumber 
\end{eqnarray}
We notice that, for every column, the first $k$ factors of the entry in the first row are factors of every row therefore we can pull them out; from the first column we pull out $(N+k)(N+k-1)\dots(N+2)(N+1)$, from the second column we pull out $(N+k+1)(N+k)\dots(N+3)(N+2)$ and so on until the last column where we pull out $(N+2k-1)\dots(N+k+1)(N+k)$. Therefore, \eqref{factdet1} equals
\begin{eqnarray}
&&\frac{(N+1)\dots(N+k-1)^{k-1}(N+k)^k(N+k+1)^{k-1}(N+k+2)^{k-2}\dots(N+2k-1)}{((2k-1)!)^k} \times \det \M,
\end{eqnarray}
where
\begin{eqnarray*}
 \det \M= \det \footnotesize{\begin{bmatrix}
N(N-1)\dots(N-k+2) & (N+1)N\dots(N-k+3) & \dots & (N+k-1)\dots (N+1) \\
(N+k+1)N\dots(N-k+3) & (N+k+2)(N+1)\dots(N-k+4) & \dots & (N+2k)(N+k-1)\dots (N+2) \\
\vdots & \vdots & \vdots & \vdots \\
(N+2k-1)\dots(N+k+1) & (N+2k)\dots(N+k+2) & \dots & (N+3k-2)\dots(N+2k-1)
\end{bmatrix}.}&&
\end{eqnarray*}
We pause here to explain the structure of the above matrix, which we will call $\M$: the first row contains the last $k-1$ original factors, since we factored out the first $k$. The second row however, will still have its first original factor followed by its $k-2$ last factors, since we factored out $k$ factors starting with the second factor. The third row will have its first 2 factors followed by its $k-3$ last factors, since we factored out $k$ factors starting from the third factor, and so on until the last row which simply has its first $k-1$ original factors since we factored out the last $k$ factors. We illustrate this as well as the row reduction process with an example where $k=4$ below. For the row reduction, the important observation is that every pair of rows $\M_j$ and $\M_{j+1}$ will have $k-2$ factors in common; the only difference will be between the last factor in $\M_j$ and the first factor in $\M_{j+1}$, and their difference will always be $2k-1$. Therefore, starting with the $k^{th}$ row, we remove the $k-1^{th}$ row from the $k^{th}$, then remove the $k-2^{th}$ row from the $k-1^{th}$, and so on, leaving the first row fixed. Now we can pull out a factor of $(2k-1)$ from every row except the first, so we pull out $(2k-1)^{k-1}$. We note that we are not changing the determinant since we are simply taking linear combinations of rows. We repeat this process, this time fixing both the first and second rows, and now the difference between entries from one row to the one above will be $2k-2$, therefore we pull out a factor of $(2k-2)^{k-2}$. By reiterating this process, each time fixing one more row, we eventually have a matrix that looks like
$$\footnotesize{\begin{bmatrix}
N(N-1)\dots(N-k+2) & (N+1)N\dots(N-k+3) & \dots & (N+k-1)\dots (N+1) \\
N\dots(N-k+3) & (N+1)\dots(N-k+4) & \dots & (N+k-1)\dots (N+2) \\
\vdots & \vdots & \vdots & \vdots \\
N & N+1 & \dots & N+k-1 \\
1 & 1 & \dots &  1
\end{bmatrix},}\nonumber$$ and we've pulled out 
\begin{eqnarray}
(2k-1)^{k-1}(2k-2)^{k-2}\dots(k+1)
\end{eqnarray} from the determinant.
Now, we will repeat this process but on the columns. We first subtract the $(k-1)^{th}$ column from the $k^{th}$, then the $(k-2)^{th}$ from the $(k-1)^{th}$, and so on fixing the first column. We repeat the process but fixing the second column, then we repeat again, each time fixing one more column. The result is a rotated upper triangular matrix, which looks like 
$$\footnotesize{\begin{bmatrix}
N(N-1)\dots(N-k+2) & (k-1)N\dots(N-k+3) & \dots & & (k-1)(k-2)\dots 2 \\
N\dots(N-k+3) & (k-2)N\dots(N-k+4) & \dots & (k-2)\dots 2 &  0 \\
\vdots & \vdots & \vdots & &\vdots \\
N & 1 & \dots & 0 & 0\\
1 & 0 & \dots & 0& 0
\end{bmatrix},}\nonumber$$
and whose determinant is 
\begin{equation}
    (-1)^{\binom{k}{2}}(k-1)(k-2)^2(k-3)^3 \dots 2^{k-2}.
\end{equation} Finally, after all this manipulation, we have that 
\begin{equation}
 \det \M = (-1)^{\binom{k}{2}}(2k-1)^{k-1}(2k-2)^{k-2}\dots(k+1)(k-1)(k-2)^2(k-3)^3 \dots 2^{k-2},
\end{equation} therefore
\begin{align}\label{niceformdet}
&\int_{U(N)} \abs{\Lambda_X(1)}^{2k}\dX =(-1)^{\binom{k}{2}}\det \left[ \binom{N+k+i+j-2}{2k -1} \right] = \\
& \frac{(N+1)\dots(N+k-1)^{k-1}(N+k)^k(N+k+1)^{k-1}\dots(N+2k-1)}{((2k-1)!)^k} \det \M =\nonumber \\
& \frac{2!\cdot3!\dots(k-1)!(N+1)\dots(N+k-1)^{k-1}(N+k)^k(N+k+1)^{k-1}\dots(N+2k-1)}{(2k-1)!(2k-2)!\dots k!}, \nonumber
\end{align}
where here we have arrived at the familiar polynomial form of the moment, i.e. equation~\eqref{eq:ks product}.

Now, with this method in mind, we will factor the determinant \eqref{niceformdet} out of \eqref{sumofdets1}. We do this by showing it is a factor of every determinant in the sum. This is true for the simple fact that $\binom{N}{a+b} = \binom{N}{a}f(N,a,b)$, where $f(N,a,b)$ is a polynomial in $N$ of order $b$.  Since $ t_j \geq 0$, we can rewrite
\begin{eqnarray}
   && \binom{N+k+i+j-2}{2k +t_j -1}  \\
   && =  \binom{N+k+i+j-2}{2k -1} \frac{(N+i+j-k-1)(N+i+j-k-2)\dots (N+i+j-k-t_j)}{(2k + t_j -1)(2k+t_j - 2) \dots (2k)}.  \nonumber 
\end{eqnarray} Therefore, every entry in the matrices of the determinant from \eqref{termtofactor} is a factor of the corresponding entry of matrices from \eqref{sumofdets1}. Therefore, we can pull out the factors of 
\begin{equation} \label{eq:idontknow}
    \frac{(N+1)\dots(N+k-1)^{k-1}(N+k)^k(N+k+1)^{k-1}(N+k+2)^{k-2}\dots(N+2k-1)}{((2k-1)!)^k}
\end{equation} as before. Equation (\ref{eq:idontknow}) is the $2k$th moment (\ref{niceformdet}) up to a constant factor.   What is left is a polynomial in $N$. It is of order  $2k$ because we know from previous literature (for example \cite{kn:akw22}, but see the introduction for the full history) that the $2k$th moment of the derivative is order $k^2+2k$ in $N$. Therefore what is left (after we factor out the order $k^2$ polynomial that is the moment (\ref{niceformdet})) is a polynomial in $N$ of order $2k$. 
\end{proof}

\subsection{Roots of $f(N,k)$}
Finally, it is interesting to note that $f(N,k)$ has small roots, which are distributed in an ellipse-like shape very near to the origin. In Figure \ref{fig:roots} we include the plots of the roots of $f(N,k)$ for small $k$. 

\begin{figure}[]
\centering
\begin{subfigure}[b]{.49\linewidth}
\includegraphics[width=\linewidth]{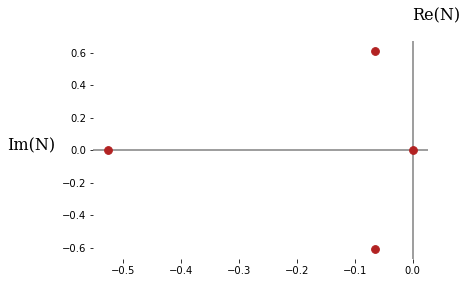}
\caption{Roots of $f(N,2)$}
\end{subfigure} 
\begin{subfigure}[b]{.49\linewidth}
\includegraphics[width=\linewidth]{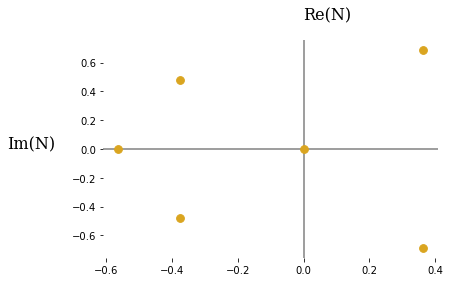}
\caption{Roots of $f(N,3)$}
\end{subfigure} 

\begin{subfigure}[b]{.49\linewidth}
\includegraphics[width=\linewidth]{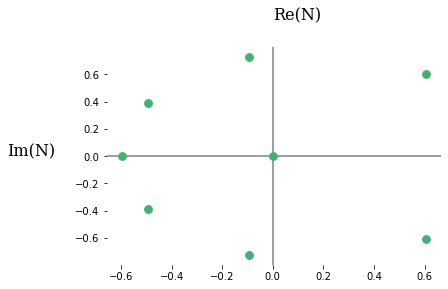}
\caption{Roots of $f(N,4)$}
\end{subfigure} 
\begin{subfigure}[b]{.49\linewidth}
\includegraphics[width=\linewidth]{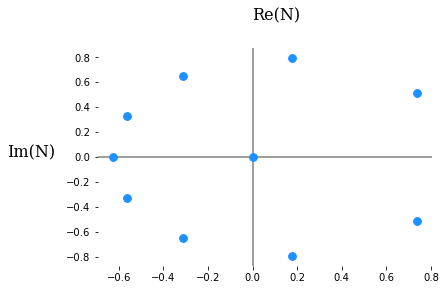}
\caption{Roots of $f(N,5)$}
\end{subfigure} 

\begin{subfigure}[b]{.49\linewidth}
\includegraphics[width=\linewidth]{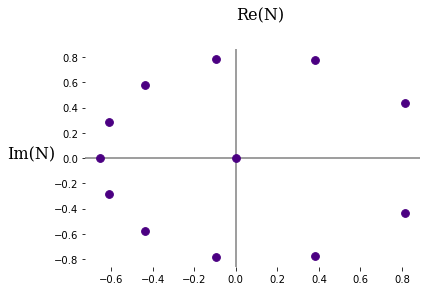}
\caption{Roots of $f(N,6)$}
\end{subfigure}  
\begin{subfigure}[b]{.49\linewidth}
\includegraphics[width=\linewidth]{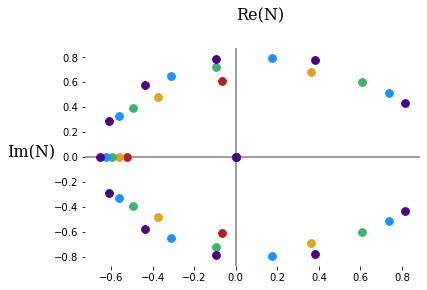}
\caption{Roots of $f(N,k)$ for $k = 2,3,4,5,6$}
\end{subfigure} 
\caption{Roots of $f(N,k)$}\label{fig:roots}
\end{figure}

In order to align the roots onto roughly the same ellipse, we find the appropriate scaling for the real roots, and then scale the complex roots accordingly. That is, since every $f(N,k)$ has two real roots, one equal to zero and the other non zero, denote the non zero real root of $f(N,k)$ by $r(k)$. Then, we find the coefficient call it $C_k$ that is the ratio of $r(k)/r(6)$; we scale by the root of $f(N,6)$ simply because it is the largest one in our data, otherwise we scale by the largest available. Then, we scale both the real and imaginary parts of the roots of each $f(N,k)$ by $C(k)$: then we see empirically in Figure \ref{fig:scaledroots} that the arguments of the zeros of $f(N,k)$s are pairwise interlacing in the upper and lower half planes.

\begin{figure}[]
\centering
\begin{subfigure}[b]{.49\linewidth}
\includegraphics[width=\linewidth]{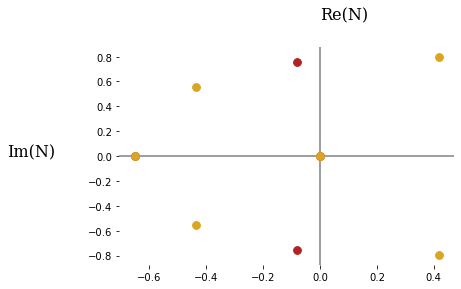}
\caption{Roots of $f(N,2)$ and $f(N,3)$}
\end{subfigure} 
\begin{subfigure}[b]{.49\linewidth}
\includegraphics[width=\linewidth]{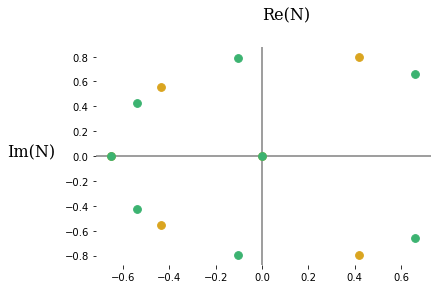}
\caption{Roots of $f(N,3)$ and $f(N,4)$}
\end{subfigure} 

\begin{subfigure}[b]{.49\linewidth}
\includegraphics[width=\linewidth]{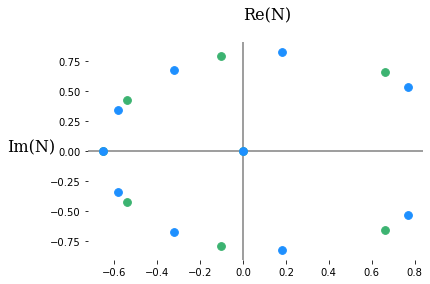}
\caption{Roots of $f(N,4)$ and $f(N,5)$}
\end{subfigure}  
\begin{subfigure}[b]{.49\linewidth}
\includegraphics[width=\linewidth]{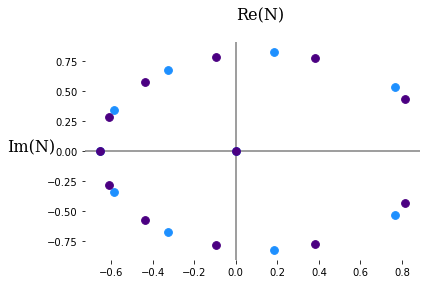}
\caption{Roots of $f(N,5)$ and $f(N,6)$}
\end{subfigure} 
\caption{Roots of $f(N,k)$ and $f(N,k+1)$}
\label{fig:scaledroots}
\end{figure}

\subsection{Moments modulo $4k-1$}\label{sect:modulo}

In the case where $4k-1$ is a prime we can determine the polynomial $f(N,k) \mod 4k-1$ explicitly.

\begin{theorem} \label{theo:mod4k-1}
For $k$ a positive integer, if $4k-1$ is prime
\begin{eqnarray}
    &&(4k-1) \int_{U(N)}|\Lambda_X'(1) |^{2k} dX =  \notag \\
    &&=  \frac{(-2) (N-2k+1) (N-2k+2)\cdots N}{ (k-1)! (k-1)!}\int_{U(N)} \abs{\Lambda_X(1)}^{2k}\dX \mod 4k-1.
\end{eqnarray}
\end{theorem}
See the introduction for a comment about the factor of $4k-1$ on the left hand side.

Before starting the proof, we develop a determinantal identity. 
\begin{lemma}\label{lem:det} We consider $k\times k$ determinants. Here $a_1, \ldots, a_k$ and $A$ are column vectors of length $k$. The notation $\det(a_1,a_2,\ldots,a_k)$ represents the determinant of a matrix that has columns $a_1$ to $a_k$. The following relation is true:
\begin{eqnarray}&& \det(A,a_2-a_1,a_3-a_2,\ldots,a_k-a_{k-1})=\det(A,a_2,a_3,\ldots a_k) + \det(a_1, A,a_3, \ldots, a_k) \nonumber \\
   &&\qquad \qquad+\det(a_1,a_2,A,a_4,\ldots,a_k) +\cdots + \det(a_1, a_2, a_3, \ldots, a_{k-1},A)
\end{eqnarray}
\end{lemma}

\begin{proof}[Proof of Lemma \ref{lem:det}]
    We will start with the $3\times 3$ example and then prove the lemma in general. Using standard properties of determinants:
    \begin{eqnarray}
        && \det(A,b,c)+\det(a,A,c)+\det(a,b,A)\nonumber \\
        &&\qquad = \det(A,b,c)+\det(A,-a,c)+\det(a,b,A)\nonumber \\
        &&\qquad =\det(A,b-a,c)+\det(a,b,A)\nonumber \\
        &&\qquad= \det(A,b-a,c)+\det(a,b-a,A)\nonumber \\
        &&\qquad= \det(A,b-a,c)+\det(A,b-a,-a)\nonumber \\
        &&\qquad = \det(A, b-a, c-a) = \det(A,b-a,c-b)
    \end{eqnarray}
    And now the same process in general:
    \begin{eqnarray}
        && \det(A,a_2,a_3,a_4,\ldots,a_k) + \det(a_1,A,a_3,a_4,\ldots,a_k) \nonumber \\
       && \qquad\qquad + \det(a_1,a_2,A,a_4,\ldots,a_k)+\cdots 
        +\det(a_1,a_2,a_3,a_4,\ldots,A) \nonumber \\
        &&=\det(A,a_2,a_3,a_4,\ldots,a_k) + \det(A,-a_1,a_3,a_4,\ldots,a_k)\nonumber \\
       && \qquad\qquad + \det(a_1,a_2,A,a_4,\ldots,a_k) +\cdots 
        +\det(a_1,a_2,a_3,a_4,\ldots,A) \nonumber \\
        &&=\det(A,a_2-a_1,a_3,a_4,\ldots,a_k) +  \det(A,a_2-a_1,-a_1,a_4,\ldots,a_k) \nonumber \\
       && \qquad\qquad+\det(a_1,a_2,a_3,A,a_5\ldots,a_k)+ \cdots 
        +\det(a_1,a_2,a_3,a_4,\ldots,A) \nonumber \\
        &&=\det(A,a_2-a_1,a_3-a_1,a_4,\ldots,a_k) \nonumber \\
       && \qquad\qquad + \det(A,a_2-a_1,a_3-a_1,-a_1,a_5,\ldots,a_k)+\cdots \nonumber\\
       && \qquad \qquad
        +\det(a_1,a_2,a_3,a_4,\ldots,A) \nonumber \\
        &&\vdots \nonumber \\
        &&=\det(A,a_2-a_1,a_3-a_1,\ldots,a_{k-1}-a_1,a_k)
        +\det(a_1,a_2,a_3,a_4,\ldots,A) \nonumber \\
         &&=\det(A,a_2-a_1,a_3-a_1,\ldots,a_{k-1}-a_1,a_k)\nonumber \\
         &&\qquad\qquad
        +\det(A,a_2-a_1,a_3-a_1,a_4-a_1,\ldots,-a_1) \nonumber \\
        &&= \det(A,a_2-a_1,a_3-a_1,\ldots,a_{k-1}-a_1,a_k-a_1)\nonumber \\
        &&=\det(A,a_2-a_1,a_3-a_2,\ldots,a_{k-1}-a_{k-2},a_k-a_{k-1})
    \end{eqnarray}
    where the final line follows by subtracting the $(k-1)$th column from the $k$th column, then the $(k-2)$th column from the $(k-1)$th, and so on. 
\end{proof}

Now we are ready to prove the main theorem of this section. 
\begin{proof}[Proof of Theorem \ref{theo:mod4k-1}]

We start with the moment of the determinant of the derivative of the characteristic polynomial written as a sum of determinants from Theorem \ref{theo:sumofdets}
\begin{eqnarray}\label{eq:fullsum}
   &&  \int_{U(N)} \abs{\Lambda_X'(1)}^{2k} \dX \nonumber \\
    && = (-1)^{\binom{k}{2}} \sum_{m=0}^k \binom{k}{m}N^{k-m} (-1)^m  \sum_{\substack{\sum_{j=1}^k t_{j} = k+m}}\binom{k+m}{t_{1}, \dots, t_{k}}    \nonumber \\ && \times 
    \det\left[ \binom{N+k+i+j-2}{2k+t_j -1} \right]_{\substack{1 \leq i \leq k \\ 1 \leq j \leq k}},
\end{eqnarray}
where the $t$ sum is over all possible non-negative integer values of the variables $t_j$, $j=1,\ldots,k$, that sum to $k+m$.

We are going to work modulo $4k-1$ with this polynomial in $N$.
Note that $k \choose m$ and $k+m \choose t_{1}, \dots, t_{k}$ are integers and do not involve $N$. However, the binomial coefficients in the determinants, when viewed as polynomials in $N$, have
rational coefficients, and some of these entries have $4k-1$ as a factor of the denominator, so some care is needed with these.
%
%
%
Specifically, if $t_{j}=2k$ for some $j$, then each entry of the $j$th column of the determinant for that term will be of the form $\binom{N+\nu}{4k-1}$ for some integer $\nu$ (where $\nu$ varies from row to row). Thus, as a polynomial in $N$,
the relevant determinant
in (\ref{eq:fullsum}) has rational coefficients that all have exactly one power of $4k-1$ in the denominator coming from 
the $(4k-1)!$ of the binomial coefficient.
Unless $t_j=2k$, a given entry
will not have any $4k-1$'s in their denominators.
Furthermore, we are assuming in this theorem that $4k-1$ is prime,
and so cannot be constructed as products of other factors. 

The condition that  $t_{j}=2k$ for some $j$ implies that we are only considering terms where $m=k$, where all the $t$ are zero except for $t_{j}=2k$. Therefore, multiplying the determinants by $((2k-1)!)^{k-1}(4k-1)!$ to clear the
denominators of the corresponding binomial coefficients,
the only terms which survive $\mod 4k-1$ are the terms on the right hand side of the following:
\small
\begin{eqnarray} \label{eq:modded}
    &&\qquad\qquad((2k-1)!)^{k-1}(4k-1)!\int_{U(N)}|\Lambda_X'(1) |^{2k} dX
    = (-1)^{\binom{k}{2}}    (-1)^k \\
    &&\times \sum_{n=1} ^k
\det \bigl[(N+k+i+j-2)(N+k+i+j-3)\dots(N-k+i+j-t_j)\bigr]_{%
  \begin{array}{l}
    1 \leq i \leq k,\\
    1 \leq j \leq k,\\
    t_j = 0, \text{ except } t_n = 2k
  \end{array}
}
\hspace{-.6in}
\mod 4k-1, \nonumber
\end{eqnarray}
\normalsize
since all the other determinants vanish $\mod 4k-1$ when multiplied by $(4k-1)!$.

As an example, when $k=2$, the various  $t$'s that occur in the sums in (\ref{eq:fullsum}) are
\begin{center}
\begin{tabular} {|c|c|}
\hline 
m&  $(t_1,t_2)$\\\hline\hline
0& (2,0)  \\
&(1,1)\\
&(0,2)\\\hline
1&(3,0)\\
&(2,1)\\
&(1,2)\\
&(0,3)\\\hline
2 &(4,0)\\
&(3,1)\\
&(2,2)\\
&(1,3)\\
&(0,4)\\\hline
\end{tabular}
\end{center}
We get a determinant in (\ref{eq:fullsum}) for every  $(t_1, t_2)$ pair. From this example, we see that the only place that we get  $t_j=4$  is when $m=2$ and $(t_1,t_2)=(4,0)$ or $(0,4)$. Thus the only determinants that will survive the process of clearing the denominator and calculating mod $4k-1=7$ are
\begin{equation}\label{eq:det11}
    \det\left[ \begin{array}{cc} (N+2)(N+1)N(N-1)(N-2)(N-3)(N-4) & (N+3)(N+2)(N+1) \\  (N+3)(N+2)(N+1)N(N-1)(N-2)(N-3) &(N+4)(N+3)(N+2)\end{array}\right] 
    \end{equation}
    and
    \begin{equation}\label{eq:det22}\det \left[\begin{array}{cc}  (N+2)(N+1)N&(N+3)(N+2)(N+1)N(N-1)(N-2)(N-3) \\  (N+3)(N+2)(N+1) &(N+4)(N+3)(N+2)(N+1)N(N-1)(N-2)\end{array}\right].
    \end{equation}

Note that because $N+3=N-4 \mod 7$ and $N+4=N-3\mod 7$, 
\begin{eqnarray}
  &&  (N+2)(N+1)N(N-1)(N-2)(N-3)(N-4)\nonumber \\
  && \qquad= (N+3)(N+2)(N+1)N(N-1)(N-2)(N-3)\mod 7\nonumber\\
  &&\qquad=(N+4)(N+3)(N+2)(N+1)N(N-1)(N-2)\mod 7
\end{eqnarray}
and thus both entries of the first column of (\ref{eq:det11}) and both entries of the second column of (\ref{eq:det22}) are all identical,
meaning that they can be factored out, leaving a column of ones. 

In general,
\small
\begin{eqnarray}
&&((2k-1)!)^{k-1}(4k-1)!\int_{U(N)}|\Lambda_X'(1) |^{2k} dX= (-1)^{\binom{k}{2}}    (-1)^k(N+2k-1)(N+2k-2)\ldots(N-2k+1) \nonumber \\
&& \qquad \times \sum_{n=1} ^k\det \left[ (N+k+i+j-2)(N+k+i+j-3)\ldots (N-k+i+j)\right]_{  \begin{array}{l}
    1 \leq i \leq k,\\
    1 \leq j \leq k,\\
    \text{column}\; n \to 1
  \end{array}
  }
\hspace{-.2in}
\mod 4k-1,
\end{eqnarray}
\normalsize
where the $n$th determinant in the sum has its $n$th column replaced by a column with 1 in every entry. 

Now we are in a setting where we can apply Lemma \ref{lem:det}. In that Lemma, we end up subtracting one column from an adjacent column, so we need the simplification that:
\begin{eqnarray}\label{eq:subtraction}
    && (N+m)(N+m-1)\cdots(N+m-2k+2)\nonumber\\
    &&\qquad \qquad-(N+m-1)(N+m-2)\cdots(N+m-2k+2)(N+m-2k+1)\nonumber \\
    &&\qquad = \big(N+m-(N+m-2k+1)\big)\times(N+m-1)\cdots(N+m-2k+2)\nonumber \\
    && \qquad =  (2k-1)(N+m-1)\cdots(N+m-2k+2).
\end{eqnarray}

So we have 
\footnotesize
\begin{eqnarray} \label{eq:deter1}
\\
&&((2k-2)!)^{k-1}(4k-1)!\int_{U(N)}|\Lambda_X'(1) |^{2k} dX= (-1)^{\binom{k}{2}}    (-1)^k (N+2k-1)(N+2k-2)\ldots(N-2k+1) \nonumber \\
&&  \times \det \left(\begin{array}{ccccc}1&(N+k)\cdots (N-k+3)&(N+k+1)\cdots(N-k+4)& \cdots & (N+2k-2)\cdots (N+1)\\ 1&(N+k+1)\cdots(N-k+4)&(N+k+2)\cdots (N-k+5)&\cdots &(N+2k-1) 
\cdots (N+2)\\\vdots&\vdots&\vdots&\ddots&\vdots\\
1& (N+2k-1)\cdots (N+2) &(N+2k)\cdots(N+3) &\cdots &(N+3k-3)\cdots (N+k)\end{array}\right)\mod 4k-1,\nonumber 
\end{eqnarray}
\normalsize

If $k=1$, the determinant is just equal to 1 and we have
\begin{equation}
    3!\int_{U(N)}|\Lambda_X'(1) |^2 dX=-(N+1)N(N-1)=-N(N-1)\int_{U(N)}|\Lambda_X(1) |^2 dX \mod 3.
\end{equation}

If $k>1$, we can return to (\ref{eq:deter1}) and subtract the $k-1$th row from the $k$th, then row $k-2$ from row $k-1$ and so forth until we have subtracted the first row from the second, noting that the subtractions collapse to a single product of consecutive integers, in exactly the same way as (\ref{eq:subtraction}):
\begin{equation}
    (N+m+1)\cdots (N+m-2k+4)-(N+m)\cdots(N+m-2k+3)=(2k-2)(N+m)\cdots(N+m-2k+4),
\end{equation}
where $m$ is an integer.

\footnotesize
\begin{eqnarray}\label{eq:reducedimension}
\\ \nonumber
&&((2k-3)!)^{k-1}(4k-1)!\int_{U(N)}|\Lambda_X'(1) |^{2k} dX= (-1)^{\binom{k}{2}}    (-1)^k (N+2k-1)(N+2k-2)\ldots(N-2k+1)\\
&&  \times \det \left(\begin{array}{ccccc}1&(N+k)\cdots (N-k+3)&(N+k+1)\cdots(N-k+4)& \cdots & (N+2k-2)\cdots (N+1)\\ 0&(N+k)\cdots(N-k+4)&(N+k+1)\cdots (N-k+5)&\cdots &(N+2k-2) 
\cdots (N+2)\\\vdots&\vdots&\vdots&\ddots&\vdots\\
0& (N+2k-2)\cdots (N+2) &(N+2k-1)\cdots(N+3) &\cdots &(N+3k-4)\cdots (N+k)\end{array}\right)\mod 4k-1.\nonumber 
\end{eqnarray}
\normalsize
Here we have pulled a factor of $(2k-2)$ from row 2 to $k$ and we note that the first column is now an identity column, effectively reducing the size of the determinant by one. 

A determinant of this form has been evaluated at (\ref{niceformdet}), once a factor of $1/(2k-1)!$ has been pulled from each of the $k$ rows of that determinant.  In (\ref{eq:reducedimension}), the matrix size is $k-1$, so we replace $k$ in (\ref{niceformdet}) by $k-1$ and replace $N$ in (\ref{niceformdet}) with $N+1$, allowing us to use that result without further modification:
\footnotesize
\begin{eqnarray}
   \\
   && (-1)^{\binom{k-1}{2}} \det \left(\begin{array}{cccc} (N+k)\cdots(N-k+4)&(N+k+1)\cdots (N-k+5)&\cdots &(N+2k-2) 
\cdots (N+2)\\\vdots&\vdots&\ddots&\vdots\\
(N+2k-2)\cdots (N+2) &(N+2k-1)\cdots(N+3) &\cdots &(N+3k-4)\cdots (N+k)\end{array}\right)\nonumber \\
\nonumber
&& \qquad =((2k-3)!)^{k-1} \frac{2\cdot 3! \cdots (k-2)!(N+2)\cdots (N+k-1)^{k-2}(N+k)^{k-1}(N+k+1)^{k-2} \cdots (N+2k-2)}{(2k-3)! (2k-4)!\cdots (k-1)!}.
\end{eqnarray}
\normalsize
Substituting this into~\eqref{eq:reducedimension} and cancelling $((2k-3)!)^{k-1}$ yields
\begin{eqnarray}\label{eq:reducedimension2}
    &&(4k-1)!\int_{U(N)}|\Lambda_X'(1) |^{2k} dX= -    (N+2k-1)(N+2k-2)\cdots (N-2k+1) \\
    &&  \times \frac{2\cdot 3! \cdots (k-2)!(N+2)\cdots (N+k-1)^{k-2}(N+k)^{k-1}(N+k+1)^{k-2} \cdots (N+2k-2)}{(2k-3)! (2k-4)!\cdots (k-1)!}.
\end{eqnarray}
In anticipation of our final formula, we gather factors on the right hand side:
\begin{eqnarray}\label{eq:tidyitup}
    &&-  \frac{2\cdot 3! \cdots (k-2)!}{(2k-3)! (2k-4)!\cdots (k-1)!} (N-2k+1) (N-2k+2)\cdots N \nonumber \\
    && \times (N+1)(N+2)^2\cdots (N+k-1)^{k-1}(N+k)^{k}(N+k+1)^{k-1} \cdots (N+2k-2)^2(N+2k-1).
\end{eqnarray}
By \eqref{eq:ks product}, the above equals
\begin{eqnarray}\label{eq:finalproduct}
    - \frac{(2k-1)!(2k-2)!(N-2k+1) (N-2k+2)\cdots N}{ (k-1)! (k-1)}
    \int_{U(N)} \abs{\Lambda_X(1)}^{2k}\dX \mod 4k-1.
\end{eqnarray}
Finally, since $4k-1$ is prime, we can use Wilson's Theorem to simplify on the left hand side: $(4k-2)! = -1 \mod 4k-1$, cancelling with the $-1$ on the right hand side and leaving $4k-1$ as a factor on the left (in the above proof, this factor of $4k-1$ cancels, ahead of reducing mod $4k-1$, with a factor of $4k-1$ that appears in the denominator of the $2k^{th}$ moment of $|\Lambda'_X(1)|$).

Also, by Wilson's Theorem, $(2k-1)!^2 = 1 \mod 4k-1$ so that, using $(2k-1)^{-1} = -2 \mod 4k-1$, we simplify on the right hand side: $(2k-1)! (2k-2)! = -2 \mod 4k-1$.

\end{proof}


\section{Determinantal formulas for moments of the derivatives $\Lambda'(1)$ and $\Lambda'(x)$}\label{sect:three}

In this section, we give several formulas for the moments of $\Lambda'(1)$ and $\Lambda'(x)$ and also describe a related differential equation satisfied by the main function that appears in the moments of $\Lambda'(1)$.

Our first theorem expresses the moments of $\Lambda'(x)$ in terms of the derivatives of a $k\times k$ determinant.

\begin{theorem}\label{thm:t1 t2}
   Let $x \in \C$, and $k$ be a non-negative integer. Then

\begin{equation}
   \label{thm:y moment cleaner}
   \int_{U(N)} \abs{\Lambda_X'(x)}^{2k} \dX \nonumber \\
    =(-1)^\frac{(k+1)k}{2} 
    \frac{d^k}{dt_1^k}
    \frac{d^k}{dt_2^k} 
    e^{-t_1 N}
    \det\left[ F_{N+k+i+j-1,k}(t_1,t_2,x) \right]_{\substack{1 \leq i \leq k \\ 1 \leq j \leq k}} \bigg|_{t_1=t_2=0} 
\end{equation}
where
\begin{equation}
    \label{eq:def F}
    F_{a,k}(t_1,t_2,x)
    =
    \frac{1}{2\pi i} \oint \frac{w^{a-1}}{(w-1)^k (w-|x|^2)^k}
    \exp\left(t_1/(w-1) + t_2/(w-|x|^2)\right) du,
\end{equation}
and the contour is a circle centred on the origin enclosing the points $1$ and $|x|^2$.
To clarify, the right hand side of the above is evaluated,
after carrying out the derivatives, at $t_1=t_2=0$.
Furthermore, if $|x| \neq 1$ and $a$ is a positive integer, then $F_{a,k}(t_1,t_2,x)$ is also equal to
\begin{align}
    \label{eq:formula for F}
    \sum_{0 \leq m+n+2k \leq a}
        \frac{t_1^m}{m!}\frac{t_2^n}{n!} &\Biggl(
            \frac{|x|^{2(a-n-k)}}{(|x|^2-1)^{m+k}}
            \sum_{l=0}^{n+k-1} {a-1 \choose n+k-1-l} { -m -k \choose l} \frac{|x|^{2l}}{(|x|^2-1)^{l}} \notag \\
            &+
            \frac{1}{(1-|x|^2)^{n+k}}
            \sum_{l=0}^{m+k-1} {a-1 \choose m+k-1-l} { -n -k \choose l} \frac{1}{(1-|x|^2)^{l}}
    \Biggr).
\end{align}
\end{theorem}
We derive the formulas in this theorem in the next subsection. We also give, in~\eqref{eq:F derivatives at 0,0}, another formula for the above expression in parentheses in terms of the $_2F_1$ hypergeometric function. 

The fact that the right hand side of this theorem depends on the norm of $x$ but not its argument again follows from the rotational invariance of Haar measure on $U(N)$.

However, we first note that,
if $|x|=1$, the theorem simplifies. To begin, setting $|x|=1$, we have
\begin{equation}
    F_{a,k}(t_1,t_2,1)
    =
    \frac{1}{2\pi i} \oint \frac{w^{a-1}}{(w-1)^{2k}}
    \exp\left((t_1+t_2)/(w-1)\right) dw.
\end{equation}
Notice that the dependence of the integrand on $t_1$ or $t_2$ appears in the exponential, and, when $x=1$, appears symmetrically. If we expand the determinant as a permutation sum whose summands are products of the entries of the matrix, we see that carrying out the differentiations with respect to $t_1$ and $t_2$ involves multiple applications of the product rule. 

Each time we differentiate a particular entry with respect to either variable, the effect is, on differentiating under the integral sign, to pull down, in the integrand, $1/(w-1)$ from the exponential. Furthermore, after carrying out all the differentiations, we set $t_1=t_2=0$.

Thus any specific multiple derivative of a given entry of the matrix in the determinant with respect to $t_1$ and $t_2$, followed by setting $t_1=t_2=0$,
can be achieved instead by letting $t=t_1+t_2$, and differentiating the determinant with respect to $t$ the same total number of times (as with respect to $t_1$ and $t_2$ combined) and setting $t=0$.

Also, note the presence of the factor $e^{-t_1 N}$ in front of the determinant which is only impacted by differentiation with respect to $t_1$, and also figures in the application of the product rule. Each differentiation with respect to $t_1$ of this factor pulls down one power of $-N$.

Hence, with $|x|=1$, we have
\begin{eqnarray}
    &&\frac{d^k}{dt_1^k}
    \frac{d^k}{dt_2^k} 
    e^{-t_1 N}
    \det\left[ F_{N+k+i+j-1,k}(t_1,t_2,1) \right]_{\substack{1 \leq i \leq k \\ 1 \leq j \leq k}} \bigg|_{t_1=t_2=0} \notag
    \\
    \label{eq:t version}
    &&= \sum_{h=0}^k {k \choose h} (-N)^{k-h} (d/dt)^{k+h}
    \det\left[ F_{N+k+i+j-1,k}(t) \right]_{\substack{1 \leq i \leq k \\ 1 \leq j \leq k}} \bigg|_{t=0},
\end{eqnarray}
where the entries in the second determinant are
\begin{equation}
    F_{a,k}(t)
    :=
    \frac{1}{2\pi i} \oint 
    \frac{w^{a-1}}{(w-1)^{2k}}
    \exp\left(t/(w-1)\right)
    dw,
\end{equation}
with the contour counter clockwise, centred on the origin and enclosing the point $u=1$.
This matches equation~\eqref{eq:moment as derivative of det}, once we include the
extra $(-1)^k$ in each summand and the  $(-1)^{k+1 \choose 2}$ in front of the right hand side of Theorem~\ref{thm:t1 t2}.

The binomial coefficient $ {k \choose h}$ arises from applying the product rule with respect to the $t_1$ variable, differentiating $k-h$ times, for $0 \leq h \leq k$. The $(-N)^{k-h}$ accounts for the number of times we differentiate $\exp(-t_1 N)$ with respect to $t_1$,
namely $k-h$.
The $(d/dt)^{k+h}$ comes about
from differentiating the first determinant $h$ times with respect to $t_1$ and $k$ times with respect to $t_2$ and, and consolidating the $t_1$ and $t_2$ variables in the entries as explained above.

Furthermore, the function $F_{a,k}(t)$ can be expressed in terms of the generalized Laguerre polynomials, whose generating function is given by:
\begin{equation}
    \frac{\exp(-zt/(1-z))}{(1 - z)^{\alpha + 1}} = \sum_{n=0}^{\infty} L_n^{(\alpha)}(t) z^n.
\end{equation}
The $ L_n^{(\alpha)}(t)$ can be hence written as a contour integral of the generating function, around a circle of radius $<1$ taken counter clockwise about $z=0$:
\begin{equation}
    L_n^{(\alpha)}(t)=
    \frac{1}{2\pi i} \oint 
    \frac{\exp(-zt/(1-z))}
    {(1 - z)^{\alpha + 1}}
    z^{-(n+1)}
    dz.
\end{equation}
Substituting $z=1/w$,
\begin{equation}
    L_n^{(\alpha)}(t)=
    \frac{1}{2\pi i} \oint 
    \frac{\exp(t/(1-w))}{(w - 1)^{\alpha + 1}}
    w^{n+\alpha} dw,
\end{equation}
where the contour is again counter clockwise along a circle centred on the origin of radius $>1$.
Therefore,
\begin{equation}
    F_{a,k}(t)
    =
    L_{a-2k}^{(2k-1)}(-t).
\end{equation}
Substituting this into \eqref{eq:t version} and then into Theorem~\ref{thm:t1 t2}, gives the following theorem.

\begin{theorem}
\label{thm:x=1}
Let $|x|=1$, i.e. on the unit circle, and $k$ be a non-negative integer. Then
\begin{eqnarray}
    &&
    \int_{U(N)} \abs{\Lambda_X'(x)}^{2k} \dX \nonumber \\
    &&\qquad  =
    (-1)^{k+1 \choose 2} \sum_{h=0}^k {k \choose h} (-N)^{k-h} (d/dt)^{k+h}
    \det_{k\times k} \left[L_{N-k+i+j-1}^{(2k-1)}(-t) \right]
    \bigg|_{t=0} \\
    && \qquad = 
    (-1)^{k} \sum_{h=0}^k {k \choose h} N^{k-h} (d/dt)^{k+h}
    \det_{k\times k} \left[L_{N+i-j}^{(2k-1)}(t) \right]
    \bigg|_{t=0}.
\end{eqnarray}
\end{theorem}
The last equality follows by reversing the columns of the determinant in the line above, introducing an extra factor of $(-1)^{k(k-1)/2}$. Additionally, replacing $-t$ with $t$ in each entry of the determinant introduces an extra
$k+h$ powers of $-1$ when we carry out the $k+h$ derivatives with respect to $t$ ahead of setting $t=0$, and these
cancel with the $k-h$ powers of $-1$ that appear in $(-N)^{k-h}$. 

In terms of implementing this formula for the purpose of calculation, the following explicit expansion of the Laguerre polynomials is handy:
\begin{equation}
    L_n^{(\alpha)}(x) 
    =
    \begin{cases} 
        \sum_{i=0}^{n} (-1)^{i} \binom{n+\alpha}{n-i} x^{i}/i! & \text{if } n \geq 0, \\
       0 & \text{if } n < 0.
    \end{cases}
\end{equation}

Interestingly, a similar formula holds in the $N$ aspect. We again allow $x \in \C$, not necessarily on the unit circle.
We first write
\begin{equation}
    \label{eq:Lambda prime identity}
    \Lambda'(x) = \Lambda(x) \frac{\Lambda'(x)}{\Lambda(x)}
    = -\Lambda(x) \sum_{j=1}^N \frac{e^{-i \theta_j}}{1-xe^{-i \theta_j}}.
\end{equation}
Therefore, averaging over $U(N)$, using the Haar measure in terms of the eigenangles,
\begin{equation}
    \label{eq:haar}
    \frac{1}{N! (2\pi)^N}\prod_{1\leq j< l\leq N} |e^{i\theta_l}-e^{i\theta_j}|^2,
\end{equation} 
and applying \eqref{eq:Lambda prime identity}, we have, for positive integer $k$ and
$x \in \C$:
\begin{eqnarray}
    &&\int_{U(N)} |\Lambda_X'(x)|^{2k} ~\dX = \notag \\
    &&
    \label{eq:N integral}
    \frac{1}{N! (2\pi)^N}
    \int_{[0,2\pi]^N}
    |\Lambda_X(x)|^{2k}
    \left(
        \sum_{j=1}^N \frac{e^{-i \theta_j}}{1-xe^{-i \theta_j}}
    \right)^k
    \left(
        \sum_{j=1}^N \frac{e^{i \theta_j}}{1-\ol{x}e^{i \theta_j}}
    \right)^k 
    \prod_{1\leq j< l\leq N} |e^{i\theta_j}-e^{i\theta_l}|^2
    \prod_{j=1}^N d\theta_j. 
\end{eqnarray}
We write
\begin{equation}
    |\Lambda_X(x)|^2 =
    \prod_{j=1}^N
    (1-xe^{-i\theta_j})
    (1-\ol{x}e^{i\theta_j})
\end{equation}
and also note, for example, as before,
\begin{equation}
    \label{eq:k power as differentiations}
    \left(
        \sum_{j=1}^N \frac{e^{-i \theta_j}}{1-xe^{-i \theta_j}}
    \right)^k
    =
    \frac{d^k}{dt_1^k}
    \exp \left(
        t_1 
        \sum_{j=1}^N \frac{e^{-i \theta_j}}{1-xe^{-i \theta_j}}
    \right)
    \bigg|_{t_1=0}.
\end{equation}
The purpose of expressing the right hand side as derivatives of an exponential
is to make the integrand more separable so as to apply Andreief's identity.
Substituting the above two identities
into~\eqref{eq:N integral}, and something similar for the conjugate
of~\eqref{eq:k power as differentiations}, we have that~\eqref{eq:N integral}
equals
\begin{multline}
    \frac{d^k}{dt_1^k}
    \frac{d^k}{dt_2^k} 
    \frac{1}{N! (2\pi)^N}
    \int_{[0,2\pi]^N}
    \prod_{j=1}^N
    (1-xe^{-i\theta_j})^k
    (1-\ol{x}e^{i\theta_j})^k\\
    \exp \left(
        \sum_{j=1}^N \frac{t_1 e^{-i \theta_j}}{1-xe^{-i \theta_j}}
        +
        \sum_{j=1}^N \frac{t_2 e^{i \theta_j}}{1-\ol{x}e^{i \theta_j}}
    \right) 
    \prod_{1\leq j< l\leq N} |e^{i\theta_j}-e^{i\theta_l}|^2
    \prod_{j=1}^N d\theta_j,
\end{multline}
evaluated at $t_1=t_2=0$. Applying the Andreief identity~\eqref{Aint} (recognizing the above
double product as a product of a Vandermonde determinant and its conjugate) to express the $N$-dimensional integral as
an $N\times N$ determinant, the above becomes
\begin{equation} 
    \frac{d^k}{dt_1^k}
    \frac{d^k}{dt_2^k} 
    \det_{N\times N}
    \left[
    \frac{1}{2\pi}
    \int_0^{2\pi}
    (1-xe^{-i\theta})^k
    (1-\ol{x}e^{i\theta})^k
    \exp \left(
        \frac{t_1 e^{-i \theta}}{1-xe^{-i \theta}}
        +
        \frac{t_2 e^{i \theta}}{1-\ol{x}e^{i \theta}}
        +i \theta(j-l)
    \right)
    d\theta
    \right] \bigg|_{t_1=t_2=0}.
\end{equation}
Specializing again to $x=1$, this becomes
\begin{equation} 
    \label{eq:andreief x=1}
    \frac{d^k}{dt_1^k}
    \frac{d^k}{dt_2^k} 
    \det_{N\times N}
    \left[
    \frac{1}{2\pi}
    \int_0^{2\pi}
    (1-e^{-i\theta})^k
    (1-e^{i\theta})^k
    \exp \left(
        \frac{t_1 e^{-i \theta}}{1-e^{-i \theta}}
        +
        \frac{t_2 e^{i \theta}}{1-e^{i \theta}}
        +i \theta(j-l)
    \right)
    d\theta
    \right]
    \bigg|_{t_1=t_2=0}.
\end{equation}
We can rewrite the first factor in the integrand as
$(1-e^{-i\theta})^k = (-1)^k e^{-i k \theta} (1-e^{i \theta})^k$.
Furthermore, $e^{-i \theta}/(1-e^{-i \theta}) = -1/(1-e^{i \theta})$,
and $e^{i \theta}/(1-e^{i \theta}) = -1 + 1/(1-e^{i \theta})$. Using these, and pulling
out $(-1)^k e^{-t_2}$ from each row of the determinant,
~\eqref{eq:andreief x=1}
equals
\begin{eqnarray}
    \label{eq:andreief x=1 b}
    &&(-1)^{kN}
    \frac{d^k}{dt_1^k}
    \frac{d^k}{dt_2^k} 
    e^{-t_2 N}
    \det_{N\times N}
    \left[
    \frac{1}{2\pi}
    \int_0^{2\pi}
    (1-e^{i\theta})^{2k}
    \exp \left(
        \frac{(t_2-t_1)}{1-e^{i \theta}}
        +i \theta(j-l-k)
    \right)
    d\theta
    \right]
    \bigg|_{t_1=t_2=0} \\
    &&=
    (-1)^{kN}
    \frac{d^k}{dt_1^k}
    \frac{d^k}{dt_2^k} 
    e^{-t_2 N}
    \det_{N\times N}
    \left[ 
        L_{j-l+k}^{(-2k-1)}(t_2-t_1)
    \right]
    \bigg|_{t_1=t_2=0}.
\end{eqnarray}
As in \eqref{eq:t version}, we apply the product rule, but with respect to $t_2$. Furthermore, we can consolidate the $t_2-t_1$ as $t$, taking care
to include a factor $(-1)^k$ to account for the effect of the chain rule each of the $k$ times we  differentiate $t_2-t_1$ with respect to $t_1$. Below,
that factor cancels with part of the $(-1)^{k-h}$ that occurs when we
differentiate $\exp(-t_2N)$ $k-h$ times with respect to $t_2$.
We thus have
\begin{theorem}
\label{thm:x=1 N version}
For positive integer $k$ 
\begin{eqnarray}
    \label{eq:}
    &&
    \int_{U(N)} \abs{\Lambda_X'(1)}^{2k} \dX \nonumber \\
    &&\qquad  =
    (-1)^{kN} \sum_{h=0}^k {k \choose h} (-1)^{h} N^{k-h} (d/dt)^{k+h}
    \det_{N\times N} \left[L_{j-l+k}^{(-2k-1)}(t) \right]
    \bigg|_{t=0}.
\end{eqnarray}
As before, this formula also holds for the average of $\abs{\Lambda_X'(x)}^{2k}$ for any $|x|=1$, by rotational
invariance of Haar measure on $U(N)$.
\end{theorem}

\subsection{Proof of Theorem~\ref{thm:t1 t2}}

As in Section \ref{sec:2}, we start with Lemma~\ref{lem:crs},
first replacing $a_j$ by $1/a_j$ for $1\leq j\leq k$, and then 
introduce the moments of $\Lambda'$ by differentiating the formula in that Lemma with respect to each of the $a_j$'s, $1 \leq j \leq 2k$. But unlike Section 2, where we then set all the $a_j$
equal to 1, here
we set $a_j=x$, and $a_{j+k}=\ol{x}$, for $1 \leq j \leq k$.

Instead of \eqref{eq:postder} we get
\begin{equation}
    \label{eq:postder x}
    x^{(N-1)}
    \left(
    \prod_{1\leq i \leq k}
    \frac{u_i}{u_i -1/x}
    \right)
    \left( N + \sum_{m=1}^k \frac{1}{1-u_m x}\right),
\end{equation}
and instead of \eqref{eq:postder2} we get
\begin{equation}
    \label{eq:postder2 x}
    -\frac{1}{\ol{x}}
    \left(
        \prod_{1 \leq i \leq k} \frac{u_i}{u_i-\ol{x}}
    \right)
    \left( \sum_{m=1}^k \frac{1}{1-u_m/\ol{x}} \right).
\end{equation}
Combining all the derivatives with respect to all the $a_j$'s gives
\begin{eqnarray}
    \label{eq:t1 t2 step1b}
    &&\int_{U(N)} |\Lambda'_X(x)|^{2k} dX =\\
    &&(-1)^k\frac{x^{kN}|x|^{-2k}}{k!(2\pi i)^k} \notag
    \oint
    \prod_{j=1}^k u_j^{N+k}
    \frac{\left(
        N
        +
        \sum_{i=1}^k
            \frac{1}
            {1-u_i x}
    \right)^k
    \left(
        \sum_{i=1}^k
            \frac{1}
            {1-u_i/\ol{x}}
    \right)^k
    }{\prod_{j=1}^k (u_i-1/x)^k (u_i-\ol{x})^k}
    \prod_{i\neq j}
    (u_i-u_j)
    \prod_{j=1}^k du_j
\end{eqnarray}
where each of the $k$-contours of this $k$-dimensional contour integral is around simple closed contours enclosing the points $1/x$ and $\ol{x}$.

We wish to make the integrand more separable so as to apply the Andreief identity, see (\ref{Aint}). To this end, we introduce parameters
$t_1$ and $t_2$, and notice that
\begin{equation}
    \frac{d^k}{dt_1^k} \exp\left(t_1 N + t_1 \sum_{i=1}^k \frac{1} {1-u_i x}  \right)
        \bigg|_{t_1=0}
    =
    \left(
    N
    +
    \sum_{i=1}^k
        \frac{1}
        {1-u_i x}
    \right)^k,
\end{equation}
and 
\begin{equation}
    \frac{d^k}{dt_2^k} \exp\left(t_2 \sum_{i=1}^k \frac{1} {1-u_i/\ol{x}}  \right)
        \bigg|_{t_2=0}
    =
    \left(
    \sum_{i=1}^k
        \frac{1}
        {1-u_i/\ol{x}}
    \right)^k.
\end{equation}
We substitute the left hand sides of these two formulas for the numerator of the displayed fraction in the integrand
in~\eqref{eq:t1 t2 step1b}.
We also note that $ \prod_{i\neq j} (u_i-u_j)$ is, up to sign, the square of a Vandermonde determinant, specifically
equal to $(-1)^{k(k-1)/2} \prod_{1 \leq j<i\leq k}(u_i-u_j)^2$. Other than this factor the
rest of the integrand separates (after carrying out the above two substitutions), and we can apply Andreief's identity (see (\ref{Aint})) to get
\begin{eqnarray}
    \label{eq:t1 t2 yipee}
    &&
    \int_{U(N)} |\Lambda'_X(x)|^{2k}  dX = \nonumber\\
    &&(-1)^{\frac{(k+1)k}{2}} \nonumber
    x^{kN}|x|^{-2k} 
    \frac{d^k}{dt_1^k}
    \frac{d^k}{dt_2^k} 
    e^{t_1 N} \\
    &&\,\,\,\,\,\,\, \det_{k \times k}\left[
        \frac{1}{2\pi i} \oint \frac{u^{N+k+i+j-2}}{(u-1/x)^k (u-\ol{x})^k}
        \exp\left(t_1/(1-ux) + t_2/(1-u/\ol{x})\right) du
    \right]_{\substack{1 \leq i \leq k \\ 1 \leq j \leq k}} \bigg|_{t_1=t_2=0}.
\end{eqnarray}
Let $a=N+k+i+j-1$, and also substitute $w=ux$ into the entry of the determinant, which thus becomes
\begin{equation}
    \label{eq:entry simplified}
    \frac{x^{2k-a}} {2\pi i}
    \oint \frac{w^{a-1}}{(w-1)^k (w-|x|^2)^k}
    \exp\left(-t_1/(w-1) - t_2 |x|^2/(w-|x|^2)\right) dw.
\end{equation}
But, in carrying out the derivatives with respect to $t_1$ and $t_2$, we can, by the chain rule,
replace $-t_1$ and $-t_2$ with $t_1$ and $t_2$ since an even number of derivatives are carried out ($k$ derivatives
with respect to each). Similarly we can replace $t_2 |x|^2$ above with $t_2$ by dropping the $|x|^{-2k}$ that appears
in~\eqref{eq:t1 t2 yipee}.

Additionally, $2k-a = -N+k-i-j+1$. We can pull out $x^{-N+k-i+1}$ from the $i$-th row of the determinant, and
$x^{-j}$ from the $j$-th column. Altogether, this pulls out, after using $\sum_{i=1}^k i = k(k+1)/2$ (and, likewise,
for $j$) and simplifying, $x^{-Nk}$ which cancels with the $x^{Nk}$ in~\eqref{eq:t1 t2 yipee}.
We have thus arrived at the first formula of Theorem~\ref{thm:t1 t2}.

Next, we derive formula~\eqref{eq:formula for F} for the function $F_{a,k}(t_1,t_2,X)$.
From~\eqref{eq:def F} we see that $F_{a,k}(t_1,t_2,X)$ an entire function of $t_1$ and $t_2$, thus 
we may expand it in a two dimensional Maclaurin series about the origin, valid for all
$t_1$ and $t_2$:
\begin{equation}
    \label{eq:F a k}
    F_{a,k}(t_1,t_2,x)
    = \sum_{m,n\geq 0} 
    F_{a,k}^{(m,n)}(0,0,x)
    \frac{t_1^m t_2^n}{m! n!},
\end{equation}
where, $F_{a,k}^{(m,n)}$ is the $m$-th derivative of $F_{a,k}(t_1,t_2,x)$ with respect to $t_1$ followed by the $n$-th derivative with respect to $t_2$. But 
\begin{equation}
   \label{eq:F m n}
    F_{a,k}^{(m,n)}(0,0,x) =
    \frac{1}{2\pi i} \oint \frac{w^{a-1}}{(w-1)^{m+k} (w-|x|^2)^{n+k}} dw.
\end{equation}
Now, as in the proof of Liouville's theorem in complex analysis, the contour intergral on the rhs vanishes if
$m+n+2k\geq a+1$, which we can see by replacing the contour with ever larger circles of radius $R$
centred on the origin, with circumference growing proportionately to $R$, whereas the integrand is at most $O(1/R^2)$.
Thus the sum over $m$ and $n$ involves finitely many terms,
and we can truncate the sum in~\eqref{eq:F a k} at $m+n+2k\leq a$. 

To evaluate the above contour integral we evaluate the residues at the points $1$ and $|x|^2$. Note that we are assuming $a$ to be a positive integer so the factor $w^{a-1}$ is entire.

We first evaluate the residue at the pole $|x|^2$. We need to determine the Laurent series for
each factor of the integrand about the point $|x|^2$. Now,
\begin{equation}
    w^{a-1} = (|x|^2+w-|x|^2)^{a-1}
    = 
    \sum_{l=0}^{a-1}
    {a-1 \choose l}
    |x|^{2(a-1-l)}
    (w-|x|^2)^l,
\end{equation}
and, if $|x| \neq 1$,
\begin{eqnarray}
    \frac{1}{(w-1)^{m+k}}
    &=&
    \frac{1}{((w-|x|^2)+(|x|^2-1))^{m+k}}
    =
    \frac{1}{(|x|^2-1)^{m+k}}
    \frac{1}{(1+(w-|x|^2)/(|x|^2-1))^{m+k}} \notag \\
    &=&
    \sum_{l=0}^\infty {-m-k \choose l}
    \frac{(w-|x|^2)^l}{(|x|^2-1)^{m+k+l}},
\end{eqnarray}
the last step being the binomial expansion for negative exponents. And, recalling complex analysis, we note, for convergence of this expansion, that when we compute the residue at a given point, we replace our contour by a small circle surrounding that point (in this case the point $|x|^2$), so that $w-|x|^2$ can be made as small as we wish on that contour.

Thus, because of the factor $(w-|x|^2)^{n+k}$ in the denominator of the integrand, we determine the residue of~\eqref{eq:F m n} at $|x|^2$ as the coefficient of $(w-|x|^2)^{n+k-1}$
in the series about $|x|^2$ of $w^{a-1}/(w-1)^{m+k}$, namely as equal to
\begin{equation}
    \label{eq:residue |x|^2}
    \frac{|x|^{2(a-1)}}{(|x|^2-1)^{m+k}}
    \sum_{l_1+l_2= n+k-1} {a-1 \choose l_1} { -m -k \choose l_2} \frac{|x|^{-2l_1}}{(|x|^2-1)^{l_2}},
\end{equation}
where $l_1$ and $l_2$ run over non-negative integers summing to $n+k-1$.

The residue at $1$ can similarly be evaluated, and equals
\begin{equation}
    \label{eq:residue 1}
    \frac{1}{(1-|x|^2)^{n+k}}
    \sum_{l_1+l_2= m+k-1} {a-1 \choose l_1} { -n -k \choose l_2} \frac{1}{(1-|x|^2)^{l_2}}.
\end{equation}
Summing the two residues and replacing $l_1$ with $n+k-1-l_2$ in the first sum, and $m+k-1-l_2$ in the second sum,
we get
\begin{eqnarray}
    \label{eq:F derivatives at 0,0}
    F_{a,k}^{(m,n)}(0,0,x)
    &=&
    \frac{|x|^{2(a-n-k)}}{(|x|^2-1)^{m+k}}
    \sum_{l=0}^{n+k-1} {a-1 \choose n+k-1-l} { -m -k \choose l} \frac{|x|^{2l}}{(|x|^2-1)^{l}} \notag \\
    &&+
    \frac{1}{(1-|x|^2)^{n+k}}
    \sum_{l=0}^{m+k-1} {a-1 \choose m+k-1-l} { -n -k \choose l} \frac{1}{(1-|x|^2)^{l}} \\
    &=& 
    {a - 1 \choose n + k - 1}
    \frac{|x|^{2(a-n-k)}}{(|x|^2-1)^{m+k}} 
    \, _2F_1\left(m + k, -n - k + 1; a - n - k + 1;|x|^2/\left(|x|^2-1\right)\right) \\
    &&+
    {a - 1 \choose m + k - 1}
    \frac{1}{(1-|x|^2)^{n+k}}
    \, _2F_1\left(n + k, -m - k + 1; a - m - k + 1;1/\left(|x|^2-1\right)\right)
\end{eqnarray}

We have thus arrived at formula~\eqref{eq:formula for F} of Theorem~\ref{thm:t1 t2}.

\subsection{An associated differential equation}
\label{subsection:diffeq}
 
Let 
\begin{equation}
    \label{eq:f_Nk}
    f_{N,k}(t): = 
    t
    \frac{
        \left(\det_{k\times k} \left[L_{N+i-j}^{(2k-1)}(t) \right]\right)'
    }
    {
        \det_{k\times k} \left[L_{N+i-j}^{(2k-1)}(t) \right]
    },
\end{equation}
i.e. $t$ times the logarithmic derivative of the displayed determinant. Then, $f_{N,k}(t)$ satisfies the differential equation.
\begin{eqnarray}
    &&
    t^2f''(t)^2+4tf'(t)^3-(4k^2-4Nt+t^2+4f(t))f'(t)^2
    \nonumber\\&&
    \qquad  -\left(2kN(2k+t)+(4N-2t)f(t)\right)f'(t)-\left(kN-f(t)\right)^2
    = 0.
    \label{eq:diffeq}
\end{eqnarray}
We initially found this differential equation for $f_{N,k}(t)$ experimentally.
In~\cite{kn:basor_et_al18}, a related function, with the entries of the determinant being $L_{N+i-j}^{(2k-1)}(-t)$ (i.e. $-t$ rather than $t$), is shown to satisfy an equivalent differential equation identified
as a $\sigma$-Painlev\'e V equation with three parameters. Specifically, 
our~\eqref{eq:diffeq} matches their equation (3-88),
with $f(t)$ equal to their $\tilde{\sigma}(t)-Nt/2$. Clarkson has pointed out to us that, while both functions satisfy the same Painlev\'e equation and initial conditions, they are different solutions. See~\cite{kn:cladun24} for a discussion about the non-uniqueness of solutions to this Painlev\'e equation.

The differential equation allows one, for example, to 
efficiently determine the moments of $\Lambda'(1)$ for
specific values of $k$ and arbitrary $N$, in comparison
to expanding the determinant in Theorem~\ref{thm:x=1} and differentiating, or summing the terms in Theorem~\ref{theo:sumofdets}.

For one, Theorem~\ref{thm:x=1} only requires us to determine the coefficients of powers of $t$ in the Maclaurin series of $\det_{k\times k} \left[L_{N+i-j}^{(2k-1)}(t) \right]$ up to terms of degree $2k$.
Writing the Maclaurin series of $f_{N,k}(t)$ up to this term as $\sum_{j=1}^{2k} c_j t^j$, with the coefficients $c_j$ depending on $N$ and $k$, the differential equation gives a recursion for the coefficients $c_j$. Note that $c_0=0$ since~\eqref{eq:f_Nk} has no constant term.

Substituting the Maclaurin series for $f_{N,k}(t)$
into the differential equation, and setting all the coefficients of powers of $t$ to 0, we successively find:
\begin{eqnarray}
    c_1 &=& -N/2 \notag \\
    c_2 &=& -{\frac {N \left( N+2\,k \right) }{4(2k-1)(2k+1)}} \notag \\
    c_3 &=& 0 \notag \\
    c_4 &=& {\frac { \left( 2N+2k+1 \right)  \left( 2N+2k-1 \right)  \left( N+2k \right) N}{ 16 \left( 2k-3 \right)  \left( 2k+3 \right)  \left( 2k-1 \right) ^{2} \left( 2k+1 \right) ^{2}}}\notag \\
    c_5 &=& 0 \notag \\
    c_6 &=& -{\frac { \left( 6{N}^{2}+12Nk+4{k}^{2}-1 \right)  \left( 2N+2k+1 \right)  \left( 2N+2k-1 \right)  \left( N+2k \right) N}{ 32 \left( 2k-5 \right)  \left( 2k+5 \right)  \left( 2k-1 \right) ^{3 } \left( 2k+1 \right) ^{3} \left( 2k-3 \right)  \left( 2k+3 \right) }} \notag \\
    c_7 &=& 0 \notag \\
    &\vdots&
\end{eqnarray}

To recover $\det_{k\times k} \left[L_{N+i-j}^{(2k-1)}(t) \right]$ from the series expansion of $f_{N,k}(t)$, we divide~\eqref{eq:f_Nk} by $t$ and then integrate with respect to $t$. This gives
\begin{equation}
   \log \det_{k\times k} \left[L_{N+i-j}^{(2k-1)}(t) \right] 
   = \sum_{j=1}^\infty c_j t^j/j + C
\end{equation}
where $C$ is the constant of integration. 
We can determine $C$ by setting $t=0$, specifically
\begin{eqnarray}
    C &=&
    \log \det_{k\times k} \left[L_{N+i-j}^{(2k-1)}(0) \right] \\
    &=& \log \det_{k\times k} \left[
    N+2k+i-j-1 \choose 2k-1 \right].
\end{eqnarray}
We can reverse the columns of the above determinant to get
\begin{eqnarray}
    C
    = \log \left((-1)^{k \choose 2} \det_{k\times k} \left[
    N+k+i+j-2 \choose 2k-1 \right] \right),
\end{eqnarray}
which we recognize as being equal to the log of the $2k$-th moment of $|\Lambda(1)|$
(see equation~\eqref{niceformdet}). This is consistent with Theorem~\ref{thm1} which includes
those moments as a factor on its right hand side. However, it is not immediately obvious from 
the differential equation that the function $f(N,k)$ is a polynomial in $N$ of degree $2k$. Rather the recursion of the differential equation only immediately yields that $f(N,k)$ is, for given $k$, a rational function of $N$.

We therefore get
\begin{equation}
   \det_{k\times k} \left[L_{N+i-j}^{(2k-1)}(t) \right] 
   =  
   \exp\left(\sum_{j=1}^\infty c_j t^j/j\right)
   \int_{U(N)} \abs{\Lambda_X(1)}^{2k}\dX.
\end{equation}
Finally, we can compose the Maclaurin series for $\exp$ with that of the exponent to get the Maclaurin series for
the determinant (with coefficients rational functions of $N$ and $k$), and substitute this into the right hand side
of Theorem~\ref{thm:x=1} to determine, for given $k$, the moments as functions of $N$.

For example, employing the  differential equation, we thus calculated the functions $f(N,7)$ and $f(N,8)$ that appear in Theorem \ref{thm1} in a few seconds. We found that
\begin{eqnarray}&&f(N,7)=N \big(395850216912899348 N^{13}+211532624477224855 N^{12}+150409183615071976
   N^{11}\nonumber\\&&\qquad+124529753766572861 N^{10}+108717221805362394 N^9+99500444626471665
   N^8\nonumber\\&&\qquad+94746015810816508 N^7+94787692493435963 N^6+100709159551410998
   N^5\nonumber\\&&\qquad+112720739347604080 N^4+123320249823386616 N^3+113230079581194576
   N^2\nonumber\\&&\qquad+70456770368487360 N+20564820256780800\big)/6249929305402823040000
   \end{eqnarray}
 and
\begin{eqnarray}&&
    f(N,8)=N \big(294731809494409081373 N^{15}+156236163525907760000
   N^{14}+111548545120295422636 N^{13}\nonumber\\&&\qquad+92004150627732094528 N^{12}+79833050367269223318
   N^{11}+72103908989822633280 N^{10}\nonumber\\&&\qquad+67069494832732475668 N^9+64329209879764227904
   N^8+63853927671987675845 N^7\nonumber\\&&\qquad+66311663923553088320 N^6+72711963466727700696
   N^5+82589042563096637568 N^4\nonumber\\&&\qquad+89406760833815044464 N^3+79567109646364454400
   N^2+47560703381244144000 N\nonumber\\&&\qquad+13348437875764992000\big)/18702760476120263262720000
\end{eqnarray}

Figure \ref{fig:78roots} shows a plot of the zeros of $f(N,7)$ and $f(N,8)$. Notice, as in our previous plots, that the arguments of the zeros are interlacing.


\begin{figure}[b]
\includegraphics{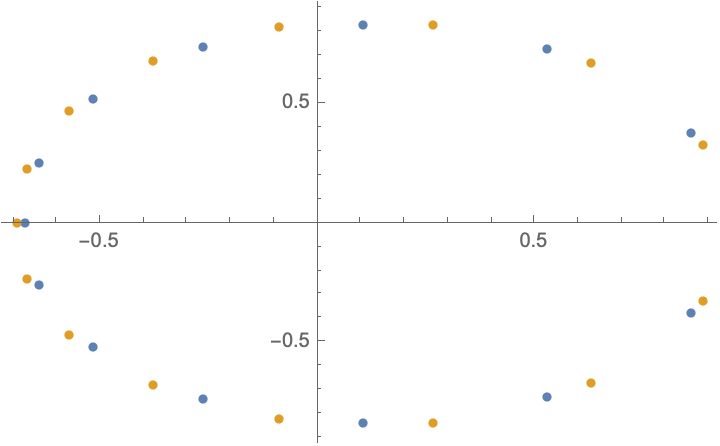}
\caption{Roots of $f(N,7)$ and $f(N,8)$}\label{fig:78roots}
\end{figure} 
 
\section{The case  $N=2$} \label{sect:four}
In the case $N=2$ we can work out more explicit formulas for the moments in terms of the $_3F_2$ hypergeometric function.

Consider
\begin{equation} \int_{U(2)} \Lambda_X'(x)^k \Lambda_{X^\dagger}'(\ol{x}) ^k \dX.
\end{equation}
For our first result in this section, we assume that $k$ is a positive integer and $x$ is complex. For a matrix $X\in U(2)$ with eigenvalues $e^{i\theta_1}$ and $e^{i\theta_2}$ we have
\begin{equation}\Lambda_X(x)=(1-xe^{-i\theta_1})(1-xe^{-i\theta_2})=1-x(e^{-i\theta_1}+e^{-i\theta_2})+x^2e^{-i(\theta_1+\theta_2)}
\end{equation}
 so that
 \begin{equation}
     \label{eq:charderiv}
     \Lambda_X'(x)=-e^{-i\theta_1}-e^{-i\theta_2}+2xe^{-i(\theta_1+\theta_2)}.
 \end{equation}
Using the fact that the joint probability density function for eigenvalues of matrices from $U(N)$ with Haar measure is 
\begin{equation}
    \frac{1}{N! (2\pi)^N}\prod_{1\leq j< k\leq N} |e^{i\theta_k}-e^{i\theta_j}|^2= \frac{1}{N! (2\pi)^N}|\Delta(e^{i\theta_1},\ldots,e^{i\theta_N})|^2,
\end{equation} 
our integral is
\begin{eqnarray}
\int_{U(2)} |\Lambda_X'(x)|^{2k} ~\dX&=&\frac{1}{8\pi^2} \int_{[0,2\pi]^2}(2x-(e^{i\theta_1}+e^{i\theta_2}))^k
(2\overline{x}-(e^{-i\theta_1}+e^{-i\theta_2}))^k |e^{i\theta_1}-e^{i\theta_2}|^2 ~d\theta_1 ~d\theta_2\nonumber \\
&=& \sum_{m,n=0}^k \binom{k}{m} \binom{k}{n} (2x)^m (2\overline{x})^n (-1)^{m+n}F(k-m,k-n)
\end{eqnarray}
where 
$$F(A,B):= \frac{1}{8\pi^2}\int_{[0,2\pi]^2} (e^{i\theta_1}+e^{i\theta_2})^A (e^{-i\theta_1}+e^{-i\theta_2})^B |e^{i\theta_1}-e^{i\theta_2}|^2 ~d\theta_1 ~d\theta_2.$$

Writing out $F(A,B)$, with binomial expansions for the quantities in brackets,
\begin{eqnarray}
    &&F(A,B)\nonumber \\
    &&=\frac{1}{8\pi^2}\int_{[0,2\pi]^2}\left(\sum_{a=0}^A \binom{A}{a}e^{i a \theta_1} e^{i(A-a)\theta_2}\right) \left(\sum_{b=0}^B \binom{B}{b}  e^{-i b\theta_1}e^{-i(B-b)\theta_2}\right)\left(2-e^{i\theta_1-i\theta_2}-e^{i\theta_2-i\theta_1}\right) d\theta_1 d\theta_2 \nonumber \\
    &&=\frac{1}{8\pi^2}\sum_{a=0}^A \sum_{b=0}^B \binom{A}{a} \binom{B}{b} \int_{[0,2\pi]^2}\left[ 2e^{i(a-b)\theta_1} e^{i(A-a-B+b)\theta_2} \right. -e^{i(a-b+1)\theta_1} e^{i(A-a-B+b-1)\theta_2} \nonumber \\
    &&\qquad \qquad\qquad\qquad\qquad\qquad\qquad\qquad\left. -e^{i(a-b-1)\theta_1}e^{i(A-a-B+b+1)\theta_2} \right]d\theta_1 d\theta_2.
\end{eqnarray}
For the $\theta_1$ integral to be non-zero, 
we require $a=b$ in the first term in the square brackets, $a=b-1$ in the second term and $a=b+1$ in the third term.
Each of these then requires $A=B$ for the $\theta_2$ integral to be non-zero.
This results in
\begin{eqnarray}
 \label{eq:F(A,A)}
 F(A,A) &=&  \sum_{a=0}^A \binom{A}{a}^2-\sum_{a=1}^A \binom{A}{a}\binom{A}{a-1}.
\end{eqnarray}
The first sum can be identified as the coefficient of $z^A$ in the squared binomial expansion of $(1+z)^A$, that is the coefficient of $z^A$ in the expansion of $(1+z)^{2A}$. So, the first sum equals $2A\choose A$. 
Furthermore, by writing ${A\choose a-1} =  {A \choose A-a+1}$, we recognize the second sum as the coefficient of $z^{A+1}$ in, as before, the square of the binomial expansion of $(1+z)^A$. Thus the second sum equals $2A \choose A+1$. However, ${2A \choose A} - {2A \choose A+1} = \frac{1}{A+1}{2A \choose A}$. In summary,
\begin{eqnarray}
F(A,B)=\left\{ \begin{array}{cc} \frac{1}{A+1} \binom{2A}{A} & \mbox{if $A=B$,} \\0 & \mbox{otherwise.}\end{array} \right.
\end{eqnarray}
Thus, we have
\begin{theorem}\label{thm:B1} For positive integer $k$ and any complex $x$ we have
\begin{eqnarray*}
\int_{U(2)} |\Lambda_X'(x)|^{2k} ~\dX&=&
\sum_{m=0}^k \frac{\binom{k}{m}^2 (4\abs{x}^2)^m\binom{2k-2m}{k-m}}{k-m+1}\\
&=&
\frac{\binom{2k}{k}}{k+1} {}_3F_2(-1-k,-k,-k;1,\frac 12 -k; |x|^2).
\end{eqnarray*}
\end{theorem}
To prove that the two right-hand sides in the statement of the theorem are equal, we use the
series for the 
 ${}_3F_2$  hypergeometric function:  
\begin{equation} \label{eq:F3def}
    {}_3F_2(a_1,a_2,a_3;b_1,b_2;z)=\sum_{n=0}^\infty \frac{(a_1)_n (a_2)_n (a_3)_n}{(b_1)_n (b_2)_n} \frac{z^n}{n!},
\end{equation}
where $(a)_n$ is the rising factorial (Pochhammer symbol) defined by 
\begin{eqnarray}
    (a)_0=1, \;\;\; (a)_n=a(a+1)\cdots(a+n-1), \; {\rm for}\;n\geq 1.
\end{eqnarray}
The idea is to compare the coefficients of like powers of $x$ in both expressions.
This comparison reduces to proving that
\begin{eqnarray}(2k)!(-1-k)_m(-k)^2_m (k-m)!^3(k-m+1)!=k!^3(k+1)!(2k-2m)!4^m(1/2-k)_m.
\end{eqnarray}
Also, it helps to turn the Pochhammer symbols into factorials:
\begin{equation}(-k)_m=(-1)^m k!/(k-m)!
\end{equation}
and
\begin{equation}
4^m (1/2-k)_m=(-1)^m \frac{(2k)!(k-m)!}{(2k-2m)! k!}.
\end{equation}

For our next result, we will also allow $k$ to be complex, but require that $|x|>1$.


For a matrix $X\in U(2)$ with eigenvalues $e^{i\theta_1}$ and $e^{i\theta_2}$ we have from (\ref{eq:charderiv}) that 
\begin{equation}
    \label{eq:with beta}
    \Lambda_X'(x)^k=
    e^{2 \pi i \beta k}
    (2x)^ke^{-ik(\theta_1+\theta_2)}
    \left(1-\frac{e^{i\theta_1}+e^{i\theta_2}}{2x }\right)^k,
\end{equation}
where the factor $\exp(2 \pi i \beta k)$ is to account, for complex exponentiation, for the fact that we
need to pay attention to the particular branch of the logarithm being used when exponentiating complex numbers with a non-integer exponent. We take the branch to be principal with argument lying
in $(-\pi,\pi]$. Therefore, the extra factor $\exp(2 \pi i \beta k)$ has $\beta \in \mathbb{Z}$,
depending on $\theta_1, \theta_2$ and $x$, selected to ensure that the imaginary part of the logarithm of the rhs above, before multiplying by $k$, lies in $(-\pi,\pi]$.
However, we will be multiplying by the $k$-th power of the conjugate $\Lambda_{X^\dagger}'(\ol{x})$ in~\eqref{eq:conjugate expression}, and that formula
similarly requires an extra factor but with opposite argument, i.e. $\exp(-2 \pi i \beta k)$.
These two factors thus cancel and we focus our attention away from it.

We have assumed that $|x|>1$; therefore we may expand the last factor into an absolutely convergent binomial series and have
\begin{equation}
    \Lambda_X'(x)^k=
    e^{2 \pi i \beta k}
    (2x)^ke^{-ik(\theta_1+\theta_2)}\sum_{m=0}^\infty \binom{k}{m}\left(-\frac{e^{i\theta_1}+e^{i\theta_2}}{2x }\right)^m.
\end{equation}
Now we use the ordinary binomial theorem and have 
\begin{equation}
    \Lambda_X'(x)^k=
    e^{2 \pi i \beta k}
    (2x)^ke^{-ik(\theta_1+\theta_2)}\sum_{m1,m2=0}^\infty \binom{k}{m_1+m_2}\frac{(m_1+m_2)!}{m_1!m_2!}(-1)^{m_1+m_2}\frac{(e^{m_1i\theta_1}e^{im_2\theta_2})}{(2x)^{m_1+m_2}}.
\end{equation}
This expression simplifies to
\begin{equation}
    \Lambda_X'(x)^k=
    e^{2 \pi i \beta k}
    (2x)^ke^{-ik(\theta_1+\theta_2)}\sum_{m1,m2=0}^\infty  \frac{\Gamma(k+1)}{\Gamma(k-m_1-m_2+1)m_1!m_2!}(-1)^{m_1+m_2}\frac{(e^{m_1i\theta_1}e^{im_2\theta_2})}{(2x)^{m_1+m_2}}.
\end{equation}
Similarly, we have
\begin{equation}
    \label{eq:conjugate expression}
    \Lambda_{X^\dagger}'(\ol{x})^k=
    e^{-2 \pi i \beta k}
    (2\ol{x})^ke^{ik(\theta_1+\theta_2)}
    \sum_{m3,m4=0}^\infty  
    \frac{\Gamma(k+1)}{\Gamma(k-m_3-m_4+1)m_3!m_4!}
    (-1)^{m_3+m_4}
    \frac{(e^{-m_3i\theta_1}e^{-im_4\theta_2})}{(2\ol{x})^{m_3+m_4}}.
\end{equation}
In preparation for computing the average over $U(2)$ we observe that if we integrate the product of the above two expressions over $[0,2\pi]^2$ we get
\begin{eqnarray}\frac{1}{(2\pi)^2} \int_{[0,{2\pi}]^2}\Lambda_X'(x)^k \Lambda_{X^\dagger}'(\ol{x}) ^k
d\theta_1d\theta_2&=&
(2|x|)^{2k} \Gamma(k+1)^2\sum_{m_1,m_2}\frac{(2|x|)^{-2m_1-2m_2}}{m_1!^2m_2!^2\Gamma(k+1-m_1-m_2)^2}\nonumber\\
&=& 4^k |x|^{2 k} \, _3F_2\left(\frac{1}{2},-k,-k;1,1;\frac{1}{|x|^2}\right).
\end{eqnarray}

Now we include the factor 
\begin{equation}\frac{|\Delta(e^{i\theta_1},e^{i\theta_2})|^2}{2}=\frac{(e^{i\theta_1}-e^{i\theta_2})(e^{-i\theta_1}-e^{-i\theta_2})}{2}=1-\frac{e^{i(\theta_1-\theta_2)}+e^{i(\theta_2-\theta_1)}}{2}
\end{equation}
in the integrand. We get
\begin{eqnarray}
2^{2 k} |x|^{2 k} \, _3F_2\left(\frac{1}{2},-k,-k;1,1;\frac{1}{|x|^2}\right)-2^{2 k-2} k^2 |x|^{2 k-2} \,
   _3F_2\left(\frac{3}{2},1-k,1-k;2,3;\frac{1}{|x|^2}\right).
   \end{eqnarray}
   Upon using (\ref{eq:F3def}) for all three of the  
   ${}_3F_{2}$ functions involved, we have
   \begin{theorem}\label{thm:b2} For all $k,x \in \C$ with  $|x|>1$ we have
   \begin{eqnarray*}
   \int_{U(2)} |\Lambda_X'(x)|^{2k} dX=2^{2 k} |x|^{2 k} \, _3F_2\left(\frac{1}{2},-k,-k;1,2;|x|^{-2}\right).
   \end{eqnarray*}
   \end{theorem}
\begin{remark} 
While our derivation was for $|x|>1$, we can extend it by continuity to $|x|=1$ when $\Re k > -1$. This is because of the standard fact that the sum defining the hypergeometric function  ${}_3F_2(a_1,a_2,a_3;b_1,b_2;z)$ converges at $z=1$ if
$\Re(b_1 + b_2 - a_1 -a_2 -a_3) > 0$. 
Here this condition reads $\Re k > -5/4$. Furthermore, on the left hand side, we can specialize to $x=1$, 
by rotational invariance. For given $k$, the integrand is bounded away from the origin (mod $2\pi$). Near 
the origin, on examining~\eqref{eq:with beta} and~\eqref{eq:conjugate expression},
we can compare the moment to the integral
$\int_{[0,2\pi]^2} (\theta_1+\theta_2)^{2k} d\theta_1 d\theta_2$, which converges and is continuous for $\Re{k} > -1$.
\end{remark}

Now let us compare our results.
Let $G_0(k,x)$ be the result of numerical integration of $|\Lambda_X'(x)|^{2k}$; let
$G_1(k,x)=\frac{\binom{2k}{k}}{k+1} {}_3F_2(-1-k,-k,-k;1,-\frac 12 -k; x^2)$; and let
$G_2(k,x)=2^{2 k} x^{2 k} \, _3F_2\left(\frac{1}{2},-k,-k;1,2;x^{-2}\right)$.
We compare the triples $(G_0,G_1,G_2)$ for various $k$ and $x$.
First of all, we observe that these are all the same if $k$ is a positive integer:
$$(k,x)=(3,5/4)\to (G_0,G_1,G_2)=(713.203,713.203,713.203)$$
$$(k,x)=(3,1/3)\to (G_0,G_1,G_2)=(14.8656,14.86,14.86)$$
Next, if $k$ is not an integer and $x>1$, then $G_0$ and $G_2$ agree
$$(k,x)=(5/4,9/5)\to (G_0,G_1,G_2)=(27.5617, 14.4-.04 i, 27.5617)$$
Finally, if $k$ is not an integer and $0<x<1$ then none of them agree:
$$(k,x)=(3/4,1/5)\to (G_0,G_1,G_2)=(1.01969, 1.0409, 1.15548 - 0.13579 i)$$
So, Carlson's theorem does {\it not} apply here, in that,
for $x>1$, the two functions $G_1$ and
$G_2$ agree for all positive integers $k$, yet do not always agree in $\Re k >0$,
implying that the needed growth conditions for Carlson's Theorem do not hold.
Indeed, the function ${}_3F_2(-1-k,-k,-k;1,-\frac 12 -k; x^2)$ seems to grow too quickly along
the negative imaginary axis.

It remains to find a formula when $k$ is not an integer and $0<x<1$.

\section{Radial distribution of the roots of $\Lambda_X'(x)$}
\label{sect:zerodist}

In this section we obtain a formula for a logarithmic  average of $\Lambda'_X(r)$ for $N=2$. 
\begin{theorem}
\label{thm:log lambda-prime}
For $0\leq r<1$ we have
\begin{eqnarray*}  \int_{U(2)} \log|\Lambda_X'(r)| ~dX
  =\frac{2 r \,
   _3F_2\left(\frac{1}{2},\frac{1}{2},\frac{1}{2};\frac{3}{2},\frac{3}{2};r^2\right)
   +r \sqrt{1-r^2} +\sin ^{-1}(r)}{\pi }-\frac{1}{2}.
\end{eqnarray*}
\end{theorem}

One interest in a logarithmic average such as this is that it is intimately connected with the distribution of the zeros of $\Lambda_X'(z)$ inside the unit circle.
This question has been studied numerically  in  \cite{kn:dffhmp}. In \cite{kn:mezzadri03},   for large matrix size,  the tails of the distribution are explicitly determined. 

This question is also connected to the distribution of the zeros of the derivative of the Riemann zeta function to the right of the half-line, which in turn is at the heart of the method   developed by Levinson when he proved that at least one-third of the zeros of the Riemann zeta-function are on the critical line.

We are able to derive our theorem from Jensen's formula
together with a calculation about the radial distribution of the roots of $\Lambda'_X(r)$.
Jensen's formula asserts that
\begin{eqnarray}
\frac{1}{2\pi} \int_0^{2\pi} \log |f(re^{i\theta})|~d\theta= \log|f(0)|+\log\frac{r^n}{|z_1\dots z_n|}
\end{eqnarray}
for any function  $f$  which is analytic in the disc $|z|\le r$, with $f(0)\neq 0$, and zeros $z_1,\dots, z_n$ inside that disc, counted with multiplicity.
The last term in this expression can be written as a Stieltjes integral as
\begin{equation}\int_0^r \log\frac r u d N_f(u)
\end{equation}
where the radial distribution of the zeros is given by
\begin{equation}N_f(u)=\sum_{z_n, f(z_n)=0\atop 
|z_n|\le u}1
\end{equation}
i.e. it is  the counting function of the zeros of $f$ in $|z|\le u$.

Jensen's formula, integrated over $U(2)$, becomes 
\begin{eqnarray}
    \label{eq:doubleavg1}
    \int_{U(2)} \frac{1}{2\pi} \int_0^{2\pi} \log |\Lambda_X'(re^{i\theta})|~d\theta ~\dX
    &=&  \int_{U(2)} \log|\Lambda_X'(0)|~\dX +
    \int_{U(2)} \int_0^r \left(\log\frac r u \right) N'_{\Lambda'}(u)~du~dX
    \notag \\
    &=& 
    \int_{U(2)} \log|\Lambda_X'(0)|~\dX +
    \int_{U(2)} \int_0^r N_{\Lambda'}(u)/u~du~dX,
\end{eqnarray}
with the last equality by integration by parts.
In (\ref{eq:doubleavg1}), the invariance of the unitary group under rotation implies we can replace the inner integrand on the left hand side with $|\Lambda_X'(r)|^{2k}$.


By \eqref{eq:charderiv}, the first integral on the right hand side is
\begin{eqnarray}
    \int_{U(2)}
    \log|\Lambda_X'(0)|~\dX
    &=&\frac{1}{4\pi^2}
    \int_{[0,2\pi]^2}
    (\log| e^{i(\theta_1 -\theta_2)} + 1|)
    \left(1-\cos(\theta_1-\theta_2)\right)
    ~d\theta_1 ~d\theta_2,
\end{eqnarray}
where we have pulled out 
$|-e^{-i\theta_1}|=1$ from the absolute values without affecting it.
The integrand is therefore a function of
$\theta_1 - \theta_2$.
But the integrand is periodic with period $2\pi$.
For any given $\theta_2$, the inner integral with respect to $\theta_1$ thus evaluates the same on substituting $\theta_1 = \theta_2+\theta$. Therefore, the right hand side above reduces to a one-dimensional integral and we have
\begin{eqnarray}
    \label{eq:1 dim}
    \int_{U(2)}
    \log|\Lambda_X'(0)|~\dX
    &=&\frac{1}{2\pi}
    \int_0^{2\pi}
    (\log| e^{i\theta} + 1|)
    \left(1-\cos(\theta)\right)
    d\theta
    \notag \\
    &=&\frac{1}{2\pi}
    \int_0^{2\pi}
    \log| e^{i\theta} + 1| d\theta
    -\frac{1}{2\pi}
    \int_0^{2\pi}
    (\log| e^{i\theta} + 1|) \cos(\theta) d\theta.
\end{eqnarray}
The first integral on the right hand side can be evaluated using Gauss' Mean Value Theorem, for the function $\log(1+re^{i\theta})$, with $r<1$, which we may do since the function $\log(1+z)$ is analytic in the unit circle,
and then taking the real part which is $\log|1+re^{i\theta}|$.
Letting $r \to 1^{-}$, this gives a value, by Gauss' Mean Value Theorem, of $\log(1)=0$ for the first integral.

For the second integral
we introduce an extra factor of $1/2$ in front of the integral so as to square the absolute value inside the logarithm, and also use $|e^{i\theta} + 1|^2 = 2+2 \cos(\theta)$.
Integrating by parts we get
\begin{eqnarray}
    &&-\frac{1}{4\pi}
    \int_0^{2\pi}
    \log( 2 + 2\cos(\theta))
    \cos(\theta)
    d\theta
    =
    -\frac{1}{4\pi}
    \int_0^{2\pi}
    \frac{\sin(\theta)^2}
    {1 + \cos(\theta)}
    d\theta
    = -1/2,
\end{eqnarray}
with the last step on replacing $\sin(\theta)^2 = 1- \cos(\theta)^2$ $= (1-\cos(\theta))(1+\cos(\theta))$, and cancelling the last factor.

Putting this together, ~\eqref{eq:1 dim} becomes
\begin{eqnarray}
    \label{eq:2nd term}
    \int_{U(2)}
    \log|\Lambda_X'(0)|~\dX
    = -1/2.
\end{eqnarray}

Next we determine, for given $u$, the average over $U(N)$ of $N_{\Lambda_X'}(u)$ so as to swap order of integration in the last integral
in~\eqref{eq:doubleavg1}.
From (\ref{eq:charderiv}) above, we see that the zero of $\Lambda_X'(z)$ for a matrix $X\in U(2)$ is at $(e^{i\theta_1}+e^{i\theta_2})/2$.
Therefore,
\begin{eqnarray}
\int_{U(2)}N_{\Lambda_X'}(u)~dX=\frac{1}{8\pi^2} 
\int\limits_{\substack{[0, 2\pi]^2 \\ |e^{i\theta_1} + e^{i\theta_2}|\le 2u}} 
|e^{i\theta_1}-e^{i\theta_2}|^2 ~d\theta_1~d\theta_2.
\end{eqnarray}
As before, we can reduce the integral on the right hand side to a one-dimensional integral
by pulling out 
$e^{i\theta_2}$ from the absolute values in the above expression
and substituting $\theta_1 = \theta_2+\theta$ while exploiting periodicity of the exponential function, so that
\begin{eqnarray}
    \int_{U(2)} N_{\Lambda_X'}(u)~dX=\frac{1}{4\pi}
    \int\limits_{\substack{[0, 2\pi] \\ |1+e^{i\theta}|\le 2u}} 
    |e^{i\theta}-1|^2 ~d\theta.
\end{eqnarray}
The integrand simplifies as $2-2\cos(\theta)$.
Furthermore,
squaring the inequality $|1+e^{i\theta}| \le 2u$, and using $\cos(\theta)^2+ \sin(\theta)^2=1$, we have
\begin{eqnarray}
    1+\cos\theta\le 2u^2,
\end{eqnarray}
or
\begin{equation}
    \cos^{-1} (2u^2-1)\le  \theta \le \pi.
\end{equation}
The requirement $\theta \le \pi$ is on account of the $\cos^{-1}$ function, but our integral is over
$[0,2\pi]$. Hence we need to include an extra factor of 2 to take into account $\theta \in (\pi,2\pi]$.
Thus, we have 
\begin{eqnarray}
\int_{U(2)}N_{\Lambda_X'}(u)~dX&=&\frac{1}{2\pi} \int_{ \cos^{-1}(2u^2-1)}^\pi (2-2\cos\theta) ~d\theta\nonumber\\
&=& \frac{2 u\sqrt{1-u^2} +\cos^{-1}\left(1-2 u^2\right)}{\pi }.
\end{eqnarray}

After swapping the order of integration of the last double integral in~\eqref{eq:doubleavg1},
we substitute the above. 
But
\begin{equation}
    \frac{2}{\pi} \int_0^r \sqrt{1-u^2} du
    = \frac{r\sqrt{1-r^2} + \sin^{-1}(r)}{\pi}.
\end{equation}
Furthermore,
\begin{equation}
    \label{eq:arccos integral}
    \frac{1}{\pi}
    \int_0^r 
    \frac{\cos^{-1}(1-2u^2)}{u} du
\end{equation}
can be expressed in terms of the series
\begin{equation}
    \cos^{-1}(1-2u^2) =
    4 \sum_{n=0}^\infty \frac{1}{2n+1} {2n \choose n} \left( \frac{u}{2} \right)^{2n+1},
\end{equation}
so that \eqref{eq:arccos integral} equals
\begin{equation}
    \frac{4}{\pi} \sum_{n=0}^\infty
    \frac{1}{(2n+1)^2} {2n \choose n} \left( \frac{r}{2} \right)^{2n+1},
\end{equation}
which can also be expressed in terms of the $_3F_2$ hypergeometric function as
\begin{equation}
   \frac{2r}{\pi}
   {}_3F_2\left(\frac{1}{2},\frac{1}{2},\frac{1}{2};\frac{3}{2},\frac{3}{2};r^2\right),
\end{equation}
thus completing the proof of the theorem.

This calculation may give a (very!) small amount of insight into what is going on in the case of general $N\times N$ matrices.

\section{Acknowledgments}

Thanks to the American Institute of Mathematics and the NSF FRG grant DMS-1854398 for supporting several visits of the authors during the course of this work.



\section{Affiliations}

Emilia Alvarez, University of Bristol, Orcid ID 0000-0003-4809-0888

Brian Conrey, American Institute of Mathematics, Orcid ID 0000-0003-3310-6736

Michael Rubinstein, University of Waterloo

Nina Snaith, University of Bristol, Orcid ID 0000-0002-0657-1918

\newcommand{\etalchar}[1]{$^{#1}$}


\end{document}